\documentclass[12pt]{article}%
\usepackage{amsmath}
\usepackage{amsfonts}
\usepackage{amssymb}
\usepackage{graphicx}
\usepackage{float}
\usepackage{tikz}
\usepackage{tkz-graph}
\usepackage{algorithm}
\usepackage{algpseudocode}
\usepackage{placeins}%
\providecommand{\U}[1]{\protect\rule{.1in}{.1in}}
\newtheorem{theorem}{Theorem}[section]

\newtheorem{corollary}[theorem]{Corollary}

\newtheorem{definition}[theorem]{Definition}
\newtheorem{example}[theorem]{Example}

\newtheorem{lemma}[theorem]{Lemma}

\newtheorem{remark}[theorem]{Remark}

\newenvironment{proof}[1][Proof]{\noindent \textbf{#1.} }{\  \rule{0.5em}{0.5em}}
\begin{document}

\title{Selfish Cops and Active Robber:\ \\Multi-Player Pursuit Evasion on Graphs}
\author{G. Konstantinidis and Ath. Kehagias \thanks{The authors thank Steve Alpern and
Pascal Schweitzer for several useful and inspiring discussions.}}
\maketitle

\begin{abstract}
We introduce and study the game of ``Selfish Cops and Active Robber'' (SCAR)
which can be seen as an multiplayer variant of the ``classic'' two-player Cops
and Robbers (CR) game. In classic CR all cops are controlled by a single
player, who has no preference over which cop captures the robber. In SCAR, on
the other hand, each of $N-1$ cops is controlled by a separate player, and a
single robber is controlled by the $N$-th player; and \emph{the capturing cop
player receives a higher reward than the non-capturing ones}. Consequently,
SCAR is an \emph{$N$-player} pursuit game on graphs, in which each cop player
has an increased motive to be the one who captures the robber. The focus of
our study is the existence and properties of SCAR \emph{Nash Equilibria} (NE).
In particular, we prove that SCAR always has one NE in deterministic
positional strategies and (for $N\ge3$) another in deterministic nonpositional
strategies. Furthermore, we study conditions which, at equilibrium, guarantee
either capture or escape of the robber and show that (because of the
antagonism between the ``selfish'' cop players) the robber may, in certain
SCAR configurations, be captured later than he would be in classic CR, or even
not captured at all. Finally we define the \emph{selfish cop number} of a
graph and study its connection to the classic cop number.

\end{abstract}

\section{Introduction\label{sec01}}

In this note we introduce and study the game of \textquotedblleft\emph{Selfish
Cops and Active Robber}\textquotedblright\ (SCAR) which can be seen as an
$N$\emph{-player \ variant} of the \textquotedblleft classic\textquotedblright%
\ two-player \emph{cops and robbers} (CR) game
\cite{bonato2011,nowakowski1983}.

The rules of SCAR are similar to those of CR: $N-1$ cops and a robber take
turns moving along the edges of an undirected finite simple connected graph;
the robber is \emph{captured} if at the end of a turn he is located in the
same vertex as one or more cops.

However each cop in SCAR is a separate player (while in CR a single player
controls all cops). Furthermore, \emph{payoffs} are quite different from those
of CR. A complete description will be given in Section \ref{sec02}; the gist
of the matter (and the SCAR\ novelty) is that \emph{in SCAR\ the capturing
cops receive a higher reward than the remaining, non-capturing cops.} As a
result, one cop's win is another cop's \emph{partial} loss (as well as the
robber's complete loss).

In other words, while in SCAR (as in CR)\ the robber will try to maximize
capture time, each cop has a motive to minimize capture time and an additional
motive for the capture to be effected by himself; depending on some game
parameters, situations will arise in which a cop will enforce a longer capture
time to ensure that he (rather than another cop)\ captures the robber. Hence
cop cooperation cannot be taken for granted (the cops are \emph{selfish}); in
this respect, SCAR differs essentially from classic CR with a \emph{team} of
$N-1$ cops chasing a single robber.

In short, SCAR is an $N$-\emph{player} (with $N\geq3)$ \emph{pursuit evasion
game on graphs, where the interests of each of the }$N$\emph{ players are in
(partial or total) conflict with those of the remaining players}. To the best
of our knowledge such games have not been previously studied.

As will become clear in the sequel, SCAR belongs to the extensively studied
family of \emph{stochastic games}\footnote{Actually all elements of SCAR are
\emph{deterministic}; the term \textquotedblleft stochastic
games\textquotedblright\ denotes a general game family which contains, as a
special case, games deterministically evolving in time.}. For two-player
stochastic games see \cite{filar1996} and for the $N$-player case see
\cite{mertens2002,vieille2002}.

SCAR can also be seen as a \emph{pursuit / evasion game played on graphs}. The
prototypical game of this family is, of course, the classic CR game introduced
in \cite{nowakowski1983,quilliot1978}; for an extensive recent overview of the
subject see the book \cite{bonato2011}. While the connection between graph
pursuit games and game theory is a natural one, relatively little has been
published on it \cite{kehagias2013,kehagias2014,kehagias2016,kehagias2016a}.
In particular, we are aware of only one previous publication (by ourselves) on
graph pursuit games involving selfish pursuers \cite{kehagias2016}.

Graph pursuit games are also related to several other research areas:
reachability games \cite{berwanger,mazala2002}, recursive games
\cite{everett1957,vieille2002}, combinatorial games (see
\cite{fraenkel1996,nowakowski1998} and especially \cite{bonato2017}) and
\emph{differential} pursuit games \cite{isaacs1965}. It is worth noting that
the idea of selfish pursuers has been occasionally (but not
extensively)\ explored in studies of differential pursuit games
\cite{foley1974,talebi2017,tynianski2017}.

The rest of the paper is organized as follows. In Section \ref{sec02} we
present the necessary preliminaries (rules, notation etc.)\ for the analysis
of the \emph{three}-player (two cops, one robber)\ SCAR. In Section
\ref{sec03} we briefly present a game theoretic formulation of a slightly
modified version of the classic CR\ game; this formulation will be useful in
the analysis of SCAR presented in later sections. In Section \ref{sec04} we
prove that three-player SCAR\ admits \emph{Nash equilibria} in deterministic
strategies and at least one of these is an equilibrium in \emph{positional}
strategies; we also prove several additional properties, regarding the
connection of SCAR capturability to the classic \emph{cop number}. In Section
\ref{sec05} we extend our results to $N$-player SCAR (with $N\geq2$), and we
also define the \emph{selfish cop number} of a graph and study its connection
to the classic cop number. We conclude, in Section \ref{sec06}, by presenting
variants and extensions of SCAR which can be the subject of future work.

\section{Preliminaries\label{sec02}}

We denote the SCAR game played by $N$ players on $G$ by $\Gamma_{N}\left(
G|s_{0},\gamma,\varepsilon\right)  $; $s_{0}$ is the \emph{initial position}
and $\gamma,\varepsilon$ are \emph{game parameters} which will be discussed in
later sections. We will sometimes simplify the notation to $\Gamma_{N}\left(
G|s_{0}\right)  $ and / or $\Gamma_{N}\left(  G\right)  $.

The main task of this section is a rigorous definition of the
\emph{three-player} SCAR\ game $\Gamma_{3}\left(  G\right)  $ (the
generalization to $\Gamma_{N}\left(  G\right)  $, the $N$-player case, will
appear in Section \ref{sec05}).

\subsection{Basics\label{sec0201}}

\textquotedblleft Iff\textquotedblright\ means \textquotedblleft if and only
if\textquotedblright. The cardinality of set $A$ is denoted by $\left\vert
A\right\vert $; the set of elements of $A$ which are not elements of $B$ is
denoted by $A\backslash B$. We use the following sets of integers:%
\[
\mathbb{N}=\left\{  1,2,3,...\right\}  ,\qquad\mathbb{N}_{0}=\left\{
0,1,2,3,...\right\}  .
\]
Given a graph $G=\left(  V,E\right)  $, for any $x\in V$, $N\left(  x\right)
$ is the \emph{neighborhood} of $x$: $N\left(  x\right)  =\left\{  y:\left\{
x,y\right\}  \in E\right\}  $; $N\left[  x\right]  $ is the \emph{closed
neighborhood} of $x$: $N\left[  x\right]  =N\left(  x\right)  \cup\left\{
x\right\}  $.

\subsection{States and Rules of $\Gamma_{3}\left(  G\right)  $\label{sec0202}}

$\Gamma_{3}\left(  G\right)  $ is played on a undirected, finite, simple
connected graph $G=\left(  V,E\right)  $. The \emph{first player} is the cop
$C_{1}$, the \emph{second player} is the cop $C_{2}$ and the \emph{third
player} is the robber $R$. Thus the \emph{player set} is $I=\left\{
C_{1},C_{2},R\right\}  $ or, for simplicity, $I=\left\{  1,2,3\right\}  $.

The game is played in \emph{turns}, numbered by the time index $t\in
\mathbb{N}_{0}$. At the zero-th turn the initial positions of the players are
given. At every subsequent turn, a single player moves. Any player can have
the first move and they play in \textquotedblleft cyclical\textquotedblright%
\ order $...\rightarrow C_{1}\rightarrow C_{2}\rightarrow R\rightarrow...$ .
The game ends if the robber is captured, that is, if at the end of a turn he
is located in the same vertex as one or more cops. Otherwise it continues indefinitely.

A \emph{game position} or \emph{game state} has the form $s=\left(
x^{1},x^{2},x^{3},p\right)  $ where $x^{n}\in V$ is the position (vertex) of
the $n$-th player and $p\in\left\{  1,2,3\right\}  $ is the number of the
player who has the next move. The set of \emph{nonterminal} states is
\[
S^{\prime}=V\times V\times V\times\left\{  1,2,3\right\}  .
\]
We will also need a \emph{terminal state} $\tau$. Hence the full state set is
$S=S^{\prime}\cup\left\{  \tau\right\}  $.

A partition of the state set can be effected as follows. Define (for each
$n\in I$) the set $S^{n}$ of states in which the $n$-th player has the next
move:\footnote{Formally speaking and given that SCAR is a stochastic game (as
noted before) at every turn all players make a move. There is however at every
turn a single player who has a non-trivial set of moves to choose among these
and this is what for reasons of brevity we mean in the current paper by the
expression "the player who has the next move".}%
\[
S^{n}=\left\{  s:s=\left(  x^{1},x^{2},x^{3},n\right)  \in S\right\}  .
\]
Then the full state set can be partitioned as follows:%
\begin{equation}
S=S^{\prime}\cup\left\{  \tau\right\}  =S^{1}\cup S^{2}\cup S^{3}\cup\left\{
\tau\right\}  . \label{eq013b}%
\end{equation}

An alternative partition of the state set is effected as follows. We define
\emph{capture state} sets:%

\begin{align*}
S_{C}^{1}  &  =\left\{  s:s=(x^{1},x^{2},x^{3},n)\text{ with }x^{1}%
=x^{3}\text{, }x^{2}\neq x^{3}\right\}  \text{, where }R\text{ is captured by
}C_{1}\text{;}\\
S_{C}^{2}  &  =\left\{  s:s=(x^{1},x^{2},x^{2},n)\text{ with }x^{1}\neq
x^{3}\text{, }x^{2}=x^{3}\right\}  \text{, where }R\text{ is captured by
}C_{2}\text{;}\\
S_{C}^{_{12}}  &  =\left\{  s:s=(x^{1},x^{2},x^{2},n)\text{ with }x^{1}%
=x^{3}\text{, }x^{2}=x^{3}\right\}  \text{, where }R\text{ is captured by both
}C_{1}\text{ and }C_{2}\text{.}%
\end{align*}
Now define:%
\begin{align*}
S_{C}  &  =S_{C}^{_{1}}\cup S_{C}^{_{2}}\cup S_{C}^{_{12}}\text{, the set of
all capture states;}\\
S_{NC}  &  =S^{\prime}\backslash S_{C}\text{, the set of all non-capture,
non-terminal states.}%
\end{align*}
Then the state set can be partitioned as follows:
\begin{equation}
S=S^{\prime}\cup\left\{  \tau\right\}  =S_{NC}\cup S_{C}\cup\left\{
\tau\right\}  . \label{eq013a}%
\end{equation}

We define $A^{n}\left(  s\right)  $, the $n$-th player's \emph{action set}
when the game state is $s=\left(  x^{1},x^{2},x^{3},m\right)  $, by%
\[
A^{n}\left(  s\right)  =\left\{
\begin{array}
[c]{ll}%
N\left[  x^{n}\right]  & \text{for }s\in S^{n}\cap S_{NC},\\
\left\{  x^{n}\right\}  & \text{for }s\in S^{m}\cap S_{NC}\text{ with }n\neq
m,\\
\left\{  \lambda\right\}  & \text{for }s\in S_{C},\text{ where }\lambda\text{
is the null move}\\
\left\{  \lambda\right\}  & \text{for }s=\tau.
\end{array}
\right.
\]
The players' action sets have the following implications on state-to-state transitions:

\begin{enumerate}
\item when the $n$-th player has the move at a non-capture state, he can stay
at his current vertex or move to any neighboring vertex, thus producing the
next state of the game;

\item when another player has the move at a non-capture state, the $n$-th
player can only stay in his current vertex (trivial move);

\item when the game is in a capture state, every player has only the null move
and the game moves to the terminal state;

\item when the game is in the terminal state, every player has only the null
move and the game moves to (actually stays in)\ the terminal state.
\end{enumerate}

\noindent These state-to-state transitions are formalized by the
\emph{transition function }$\mathbf{T}\left(  s,a\right)  $ which denotes the
game state resulting when the game is at a position $s\in S$ and each player
$n\in I$ makes a move $a^{n}\in A^{n}(s)$ resulting to the actions profile
$a=(a^{1},a^{2},a^{3})$. The behavior of $\mathbf{T}\left(  s,a\right)  $ is
illustrated by some examples as follows:
\begin{align*}
\text{for }  &  s=\left(  x^{1},x^{2},x^{3},1\right)  \in S^{1}\cap S_{NC} &
:\;  &  \mathbf{T}\left(  s,(a^{1},a^{2},a^{3})\right)  =\left(  a^{1}%
,x^{2},x^{3},2\right)  ,\\
\text{for }  &  s=\left(  x^{1},x^{2},x^{3},2\right)  \in S^{2}\cap S_{C} &
:\;  &  \mathbf{T}\left(  s,(\lambda,\lambda,\lambda)\right)  =\tau,\text{ }\\
\text{for }  &  s=\tau\text{ } & :\;  &  \mathbf{T}\left(  \tau,(\lambda
,\lambda,\lambda)\right)  =\tau.
\end{align*}
Often we use the following simplified notation: if at $t$ the game state is
$s_{t}=\left(  x_{t}^{1},x_{t}^{2},x_{t}^{3},n\right)  \in S^{n}$ and at $t+1$
the $n$-th player's action is $a_{t+1}^{n}$, we write
\[
s_{t+1}=\mathbf{T}\left(  s_{t},a_{t+1}^{n}\right)  .
\]

\subsection{Histories and Strategies\label{sec0203}}

A \emph{game history} is a sequence $h=\mathbf{(}s_{0},s_{1},s_{2},...)$,
where $s_{t}$ is the state of the game at the $t$-th turn ($t\in\mathbb{N}%
_{0}\mathbb{=}\left\{  0,1,2,3,...\right\}  $). We define the following
history sets:%
\begin{align*}
\text{histories of length }k  &  :H_{k}=\{h=(s_{0},s_{1},s_{2},...,s_{k-1}%
)\},\\
\text{histories of finite length}  &  :H_{\ast}=\cup_{k=1}^{\infty}H_{k},\\
\text{histories of infinite length}  &  :H_{\infty}=\{h=(s_{0},s_{1}%
,...,s_{t},...)\}.
\end{align*}
Histories of infinite length can be further partitioned as $H_{\infty}%
=H_{C}\cup H_{NC}$ where
\begin{align*}
\text{histories where capture occurs}  &  :H_{C}=\{h=\left(  s_{0}%
,s_{1},...\right)  \in H_{\infty}:\exists s_{t}\in S_{C}\}\text{,}\\
\text{histories where the robber evades indefinitely}  &  :H_{NC}=\{h=\left(
s_{0},s_{1},...\right)  \in H_{\infty}:\nexists s_{t}\in S_{C}\}.
\end{align*}
Furthermore, we define the \emph{capture time of history }$h\in H_{\infty}$ as%
\[
T_{C}(h)=\left\{
\begin{array}
[c]{cc}%
\min\left\{  t:s_{t}\in S_{C}\right\}  & \text{if }h\in H_{C}\\
\infty & \text{if }h\in H_{NC}\text{.}%
\end{array}
\right.
\]
We will often use the simpler notation $T_{C}$, when the history is clear from
the context. Now consider the following cases.

\begin{enumerate}
\item If $T_{C}=0$ then the initial state is a capture state and $s_{t}=\tau$
for every $t\in\mathbb{N}=\left\{  1,2,...\right\}  $.

\item If $0<T_{C}<\infty$ then:

\begin{enumerate}
\item at the $0$-th turn the game starts at some preassigned state $s_{0}\in
S_{NC}$;

\item at the $t$-th turn (for $t\in\left\{  1,2,...,T_{C}-1\right\}  $), the
game moves to some state $s_{t}\in S_{NC}$;

\item at the $T_{C}$-th turn the game moves to some capture state $s_{T_{C}%
}\in S_{C}$;

\item at the $\left(  T_{C}+1\right)  $-th turn the game moves to the terminal
state and stays there for all subsequent turns (for every $t>T_{C}$,
$s_{t}=\tau$ and the game effectively ends at time $T_{C}$).\ 
\end{enumerate}

\item If $T_{C}=\infty$ then $s_{t}\in S_{NC}$ for every $t\in\mathbb{N}_{0}$.
\end{enumerate}

\noindent The case $s_{0}=\tau$ is uninteresting and hence excluded from consideration.

In Game Theory a \emph{strategy} is a rule which prescribes (perhaps
probabilistically) a player's next move in every instance he might have to
move. A \emph{pure, or deterministic }strategy is a function $\sigma
^{n}:H_{\ast}\rightarrow A^{n}$ which assigns a move in the player's action
set to each finite-length history. That is,%
\[
\forall h=(s_{0},s_{1},s_{2},...,s_{t})\in H_{\ast}\text{ }\exists a^{n}\in
A^{n}(s_{t}):\sigma^{n}\left(  h\right)  =a^{n}\text{.}%
\]
It is well known that there is no loss of generality in assuming that each
player at the start of the game selects a \emph{deterministic} strategy
$\sigma^{n}$\emph{ }which determines all his subsequent moves (this is so
because $\Gamma_{3}\left(  G\right)  $\ is a game of \emph{perfect
information}, i.e., at each turn only one player moves and he knows all the
preceding moves). We will only consider \emph{legal} strategies, that is
strategies which never produce moves outside the player's closed neighborhood.

We call a strategy $\sigma^{n}$ \emph{positional }(or \emph{Markovian}
\emph{stationary}) if the next move depends only on the current state of the
game (but not on previous states or current time). That is,%
\[
\forall h=(s_{0},s_{1},s_{2},...,s_{t})\in H_{\ast}:\sigma^{n}\left(
h\right)  =\sigma^{n}\left(  s_{t}\right)  \text{.}%
\]

A \emph{strategy profile} is a triple $\sigma=\left(  \sigma^{1},\sigma
^{2},\sigma^{3}\right)  $. We define $\sigma^{-n}=\left(  \sigma^{m}\right)
_{m\in I\backslash\left\{  n\right\}  }$; for instance, $\sigma^{-1}=\left(
\sigma^{2},\sigma^{3}\right)  $. If the strategies $\sigma^{1},\sigma
^{2},\sigma^{3}$ are deterministic, we call $\left(  \sigma^{1},\sigma
^{2},\sigma^{3}\right)  $ a \emph{deterministic profile}. A \emph{positional}
\emph{profile} is defined analogously. If the deterministic profile $\left(
\sigma^{1},\sigma^{2},\sigma^{3}\right)  $ applied to the game $\Gamma
_{3}\left(  G|s_{0},\gamma,\varepsilon\right)  $ results (resp. does not
result)\ in a capture, we call $\left(  \sigma^{1},\sigma^{2},\sigma
^{3}\right)  $ a \emph{capturing} (resp. \emph{non-capturing}) profile in
$\Gamma_{3}\left(  G|s_{0},\gamma,\varepsilon\right)  $.

\subsection{Payoffs\label{sec0204}}

To complete the description of $\Gamma_{3}\left(  G\right)  $, we must define
payoff functions for all players. Each player will try to maximize his payoff;
the payoffs will encapsulate the following facts.

\begin{enumerate}
\item The longer the capture time, the less the cops gain and the less the
robber loses.

\item The capturing cop gains at least as much as the non-capturing one.
\end{enumerate}

Each player's \emph{payoff function} depends, in general, on the entire game
history; but, since a history is fully determined by the initial position
$s_{0}=\left(  x^{1},x^{2},x^{3},p\right)  $ and the strategy profile
$\sigma=\left(  \sigma^{1},\sigma^{2},\sigma^{3}\right)  $, we will write the
$n$-th player's payoff in any one of the equivalent forms $Q^{n}\left(
s_{0},s_{1},s_{2},...\right)  $, $Q^{n}\left(  s_{0},\sigma\right)  $ and
$Q^{n}\left(  s_{0},\sigma^{1},\sigma^{2},\sigma^{3}\right)  $.

Fix a constant $\varepsilon\in\left[  0,\frac{1}{2}\right]  $ and define
\emph{turn payoffs} as follows. For $n\in\left\{  1,2\right\}  $, $C_{n}$'s
payoff is%
\begin{equation}
q^{n}\left(  s\right)  =\left\{
\begin{array}
[c]{rl}%
1-\varepsilon & \text{if }s\in S_{C}^{n},\\
\varepsilon & \text{if }s\in S_{C}^{m}\text{ with }n\neq m,\\
\frac{1}{2} & \text{if }s\in S_{C}^{12},\text{ }\\
0 & \text{else.}%
\end{array}
\right.  \label{eq020402}%
\end{equation}
$R$'s turn payoff is
\begin{equation}
q^{3}\left(  s\right)  =\left\{
\begin{array}
[c]{rl}%
-1 & \text{if }s\in S_{C},\\
0 & \text{else.}%
\end{array}
\right.  \label{eq020401}%
\end{equation}
Next, we fix a \emph{discounting factor} $\gamma\in(0,1)$ and, for
$n\in\left\{  1,2,3\right\}  $, we define the $n$-th player's \emph{total
payoff} function by%
\begin{equation}
Q^{n}\left(  s_{0},s_{1},s_{2},...\right)  =\sum_{t=0}^{\infty}\gamma^{t}%
q^{n}\left(  s_{t}\right)  , \label{eq020403}%
\end{equation}
In the rest of the paper we will assume, unless explicitly stated otherwise,
that
\[
\left(  \gamma,\varepsilon\right)  \in\Omega^{3}=\left(  0,1\right)
\times\left[  0,\frac{1}{2}\right]  .
\]

To understand the consequences of (\ref{eq020402})-(\ref{eq020403}), let us
first consider the case where (i)$\ T_{C}<\infty$ (finite capture time), (ii)
$\varepsilon<\frac{1}{2}$ and (iii) capture is effected by a single cop. Then
the players receive the following payoffs:

\begin{enumerate}
\item the capturing cop receives $\left(  1-\varepsilon\right)  \gamma^{T_{C}%
}$;

\item the non-capturing cop receives $\varepsilon\gamma^{T_{C}}$;

\item the robber receives $-\gamma^{T_{C}}$ (i.e., loses $\gamma^{T_{C}}$).
\end{enumerate}

\noindent We see that the total cops' reward is equal to the robber's loss,
but it is not divided equally between the two cops: since $\varepsilon
<\frac{1}{2}$, the capturing cop receives more than the other one (unless both
cops simultaneously capture the robber). Furthermore, since $\gamma\in\left(
0,1\right)  $, the robber's loss is a decreasing function of capture time
$T_{C}$ and he will play so as to maximize $T_{C}$. Conversely, the cops have
a motive to minimize $T_{C}$. But, since $\left(  1-\varepsilon\right)
\gamma^{T_{C}}>\varepsilon\gamma^{T_{C}}$, there is an additional motive for
each cop to be the capturing one; there will exist combinations of
$\gamma,\varepsilon$ and $T_{C}$ for which a cop may choose to delay the
robber capture in order to ensure that it is effected by himself (an example
will be given in Section \ref{sec0403}).

Let us also consider briefly several additional scenaria which are obtained
for particular values of $T_{C},\gamma,\varepsilon$.

\begin{enumerate}
\item If the robber can avoid capture \emph{ad infinitum}, i.e., if
$T_{C}=\infty$, then all players receive zero payoff. Clearly this is the best
outcome for $R$.

\item Since a single player moves at each turn, it is only possible to have a
\textquotedblleft double capture\textquotedblright\ if $R$, on his turn, moves
into a vertex which is occupied by both $C_{1}$ and $C_{2}$. While the rules
of SCAR do not exclude this possibility, clearly it can only result if $R$
plays suboptimally. In this case each cop will receive equal payoff of
$\gamma^{T_{C}}/2$.

\item When $\varepsilon=\frac{1}{2}$ (and for any $\gamma\in\left(
0,1\right)  $) each cop receives the same payoff whether he captures $R$ or
not; hence one might expect the two cops to collaborate to effect capture in
the shortest possible time as in classic CR\ played by two cops against one
robber; however this is not always the case, as we shall see in Section
\ref{sec04}.
\end{enumerate}

The definition of payoffs completes the description of $\Gamma_{3}\left(
G|s,\gamma,\varepsilon\right)  $; it is easily recognized that it is a
\emph{discounted stochastic game} \cite{filar1996}: (i) it consists of a
sequence of one-shot games, where the game played at any time depends on the
previous game played and the actions of the players; (ii)\ the total payoff is
the time-discounted sum of the one-shot games payoffs.

\subsection{Equilibria\label{sec0205}}

In the classic CR\ game, a basic question is the existence of \emph{winning}
and/or \emph{(time)\ optimal} strategies. As we will see later, SCAR\ does not
necessarily possess optimal strategies; but one can still look for
\emph{equilibrium strategy profiles}, in which no player has a motive to
\emph{unilaterally} change his strategy. The prevalent definition of
equilibrium is the one due to Nash \cite{nash1950}, which we now present in
the general context of $N$-player stochastic games; the application to
three-player SCAR\ (and $N$-player SCAR, as we will see in Section
\ref{sec05}) is immediate.

Consider an $N$-player perfect-information stochastic game starting at state
$s_{0}$. When the players use the (deterministic) strategy profile
$\sigma=\left(  \sigma^{1},\sigma^{2},...,\sigma^{N}\right)  $ they receive
(total) payoffs $Q^{1}\left(  s_{0},\sigma\right)  $, $Q^{2}\left(
s_{0},\sigma\right)  $, ..., $Q^{N}\left(  s_{0},\sigma\right)  $. We say that
$\sigma_{\ast}=\left(  \sigma_{\ast}^{1},\sigma_{\ast}^{2},...,\sigma_{\ast
}^{N}\right)  $ is a \emph{Nash equilibrium} (NE)\ iff%
\begin{equation}
\forall n,\forall\sigma^{n}:Q^{n}\left(  s_{0},\sigma_{\ast}\right)  \geq
Q^{n}\left(  s_{0},\sigma^{n},\sigma_{\ast}^{-n}\right)  . \label{eq03001}%
\end{equation}
What (\ref{eq03001}) says is that, when the rest of the players stick to their
equilibrium strategies, no player can improve his payoff by \emph{unilaterally
}changing his own; for example, if players $2,3,...,N$ play $\sigma_{\ast
}^{-1}=\left(  \sigma_{\ast}^{2},...,\sigma_{\ast}^{N}\right)  $, then the
first player cannot increase his payoff by switching from $\sigma_{\ast}^{1}$
to some other $\sigma^{1}$.\footnote{The definition of NE can be extended to
games of \emph{non}-perfect information, provided the $\sigma^{n}$'s are
understood as \emph{probabilistic} strategies and the $Q^{n}$'s as
\emph{expected} payoffs. We will not need these generalizations in the current
paper.} The following points must be emphasized.

\begin{enumerate}
\item A game may possess no NE, or exactly one, or more than one.

\item A NE is a \emph{strategy} profile; different NE may yield the same
\emph{payoffs} to the players.

\item Different NE may yield different payoffs. The fact that $\sigma_{\ast}$
is a NE does not imply that the corresponding payoff is the best a player can
achieve; if \emph{more than one }players change their strategies, they may
achieve better payoffs than the ones implied by a NE. In other words, a NE is
not necessarily an optimal solution.
\end{enumerate}

\section{Modified CR from a Game Theoretic Point of View\label{sec03}}

Before embarking on the study of SCAR, it will be useful to present a game
theoretic formulation of the following, modified CR game.

\begin{enumerate}
\item The game is played \ by two players:\ the \emph{cop player} controls
$N-1$ cop \emph{tokens} (with $N\geq2$) and the \emph{robber player} controls
a single robber token.

\item States, movement rules, histories and capture time are defined in
exactly the same way as those of SCAR.

\item The same is true for strategies except for the fact that the cop
player's strategy is of the form $\left(  \sigma^{1},...,\sigma^{N-1}\right)
$, i.e., it contains one strategy for each of his tokens.

\item The cop (resp. robber)\ player's payoff is $\gamma^{T_{C}}$ (resp.
$-\gamma^{T_{C}}$); if $T_{C}=\infty$, both players receive zero payoff.
\end{enumerate}

\noindent It can be easily seen that this is a \emph{two-player, zero-sum,
discounted }stochastic game. It only differs from the classic CR game, played
with $N-1$ cops and one robber, in the following. \ 

\begin{enumerate}
\item In modified CR, the cop (resp. robber)\ player tries to maximize (resp.
minimize)\ $\gamma^{T_{C}}$; this is obviously equivalent to classic CR, where
the cop (resp. robber)\ player tries to minimize (resp. maximize)\ $T_{C}$.

\item In modified CR, time is counted in turns, while in classic CR\ it is
counted in \emph{rounds}, where each round consists of one move for each
player. This is roughly equivalent to a rescaling of time by the factor $1/N$.

\item In classic CR, the players \emph{select} their initial positions, while
in modified CR the initial position is \emph{predetermined}. It is easy to
recover this aspect of classic CR\ by adding to modified CR a
\textquotedblleft placement turn\textquotedblright\ for each player; this
change would have no major impact on the essential features of the game such
as the existence of value and optimal strategies.
\end{enumerate}

Using standard results \cite[Section 4.3]{filar1996} we see that modified CR
has a \emph{value}, which in fact is proportional to the logarithm of capture
time, and both players have\emph{ optimal positional }strategies.

As already mentioned, we can assume that the cop (resp. robber)\ player tries
to minimize (resp. maximize) the capture time. Let $T_{N}\left(
G|s_{0}\right)  $ be the capture time when the cop player has $N-1$ tokens and
both players play optimally (note the dependence on the initial position
$s_{0}$). Hence, if $T_{N}\left(  G|s_{0}\right)  $ is finite (resp.
infinite)\ then the cop's (resp. robber's)\ optimal strategies are winning
(for the respective player).

We denote the maximum value of optimal capture time over all starting
positions by%
\[
T_{N}\left(  G\right)  =\max_{s_{0}}T_{N}\left(  G|s_{0}\right)  .
\]
Assuming the game is played with $N-1$ cops and one robber, it is easily seen that:

\begin{enumerate}
\item if $T_{N}\left(  G\right)  <\infty$, then the cop player has an
(optimal) winning strategy for \emph{every} starting position;

\item for every starting position, assuming subsequent optimal play by the cop
player (but not necessarily by the robber player), the capture time is less
than or equal to $T_{N}\left(  G\right)  $.
\end{enumerate}

The \emph{cop number} of a graph $G$ is denoted by $c\left(  G\right)  $ and
defined to be the smallest number of cop tokens which guarantees finite
capture time (i.e., one less than the smallest $N$ for which $T_{N}\left(
G\right)  <\infty$). We call $G$ \emph{cop-win} if any of the following
equivalent conditions holds:

\begin{enumerate}
\item capture time is finite for CR on $G$ with one optimally played cop token;

\item $T_{2}\left(  G\right)  $ $<\infty$;

\item $c\left(  G\right)  =1$.
\end{enumerate}

In light of the above remarks, it is clear that all essential aspects of the
classic CR are captured by \ the modified CR. \emph{In the rest of the paper,
the term \textquotedblleft CR\ game\textquotedblright\ will denote the
modified game} (unless we specifically use the term \textquotedblleft classic
CR\textquotedblright).

Finally note that, when $N=3$ (two cops vs. one robber) the modified two-cops
CR\ is \emph{path-equivalent} to $\Gamma_{3}\left(  G\right)  $, by which we
mean the following. Take any strategies $\sigma^{1},\sigma^{2},\sigma^{3}$
and, starting from the same position $s_{0}$, apply them to:

\begin{enumerate}
\item $\Gamma_{3}\left(  G\right)  $, with $\sigma^{n}$ being the strategy of
the $n$-th player;

\item the modified two-cops CR, with $\sigma^{1}$\ (resp. $\sigma^{2}$)\ being
the strategy the cop player uses for his first (resp. second) token, and
$\sigma^{3}$ being the strategy the robber player uses.
\end{enumerate}

\noindent Then it is clear that the same infinite history will be produced in
both games.

\section{Three-Player SCAR\label{sec04}}

In this section we study $\Gamma_{3}\left(  G\right)  $ and prove that it
always has both positional and non-positional NE; we also study the connection
between classic cop number and existence of capturing NE.

\subsection{Existence of a Positional NE\label{sec0401}}

First we prove the existence of at least one positional NE\ in
\emph{deterministic} strategies for $\Gamma_{3}\left(  G\right)  $.

\begin{theorem}
\label{prp0401}For every graph $G$ and for every $s_{0}\in S$, $\left(
\gamma,\varepsilon\right)  \in\Omega^{3}$ the game $\Gamma_{3}\left(
G|s_{0},\gamma,\varepsilon\right)  $ has a deterministic positional NE. More
specifically, there exists a deterministic positional profile $\sigma_{\ast
}=\left(  \sigma_{\ast}^{1},\sigma_{\ast}^{2},\sigma_{\ast}^{3}\right)  $ such
that%
\begin{equation}
\forall n,\forall s_{0},\forall\sigma^{n}:Q^{n}\left(  s_{0},\sigma_{\ast}%
^{n},\sigma_{\ast}^{-n}\right)  \geq Q^{n}\left(  s_{0},\sigma^{n}%
,\sigma_{\ast}^{-n}\right)  . \label{eq02011}%
\end{equation}
For every $s$ and $n$ let $u^{n}\left(  s\right)  =Q^{n}\left(  s,\sigma
_{\ast}\right)  $. Then the \ following equations are satisfied%
\begin{align}
\forall n,\forall s  &  \in S^{n}:\sigma_{\ast}^{n}\left(  s\right)  =\arg
\max_{a^{n}\in A^{n}\left(  s\right)  }\left[  q^{n}\left(  s\right)  +\gamma
u^{n}\left(  \mathbf{T}\left(  s,a^{n}\right)  \right)  \right]
,\label{eq02012}\\
\forall n,m,\forall s  &  \in S^{n}:u^{m}\left(  s\right)  =q^{m}\left(
s\right)  +\gamma u^{m}\left(  \mathbf{T}\left(  s,\sigma_{\ast}^{n}\left(
s\right)  \right)  \right)  . \label{eq02013}%
\end{align}

\end{theorem}

\begin{proof}
Fink has proved in \cite{fink1963} that every $N$-player discounted stochastic
game has a positional NE\ in \emph{probabilistic} positional strategies.
Fink's results apply to the general game, with \emph{simultaneous} moves by
all players and probabilistic strategies and state transitions; he proves that
the following equations\footnote{We have adapted Fink's notation to our own,
so as to fit the $\Gamma_{3}\left(  G\right)  $ context.} must be satisfied at
equilibrium for all $m\in\left\{  1,2,3\right\}  $ and $s\in S$:
\begin{equation}
\mathfrak{u}^{m}\left(  s\right)  =\max_{\pi^{m}\left(  s\right)  }\sum
_{a^{1}\in A^{1}\left(  s\right)  }\sum_{a^{2}\in A^{2}\left(  s\right)  }%
\sum_{a^{3}\in A^{3}\left(  s\right)  }\pi^{1}\left(  a^{1}|s\right)  \pi
^{2}\left(  a^{2}|s\right)  \pi^{3}\left(  a^{3}|s\right)  \left[
q^{m}\left(  s\right)  +\gamma\sum_{s^{\prime}}\Pr\left(  s^{\prime}%
|s,a^{1},a^{2},a^{3}\right)  \mathfrak{u}^{m}\left(  s^{\prime}\right)
\right]  , \label{eq02015}%
\end{equation}
where

\begin{enumerate}
\item $\mathfrak{u}^{m}\left(  s\right)  $ is the expected value of
$u^{m}\left(  s\right)  $;

\item $\pi^{m}(a^{j}|s)$ is the probability that, given the current state is
$s$, the $m$-th player plays $a^{j}$;

\item $\pi^{m}(s)=\left(  \pi^{m}(a^{m}|s)\right)  _{a^{m}\in A^{m}\left(
s\right)  }$ is the vector of all probabilities (i.e., for all available actions);

\item $\Pr\left(  s^{\prime}|s,a^{1},a^{2},a^{3}\right)  $ is the probability
that the next state is $s^{\prime}$, given the\ current state is $s$ and the
players actions $a^{1},a^{2},a^{3}$.
\end{enumerate}

\noindent Now choose any $n$ and any $s\in S^{n}$. For all $m\neq n$, the
$m$-th player has a single move, i.e., we have $A^{m}\left(  s\right)
=\left\{  a^{m}\right\}  $, and so $\pi^{m}(a^{m}|s)=1$. Also, since
transitions are deterministic,%
\[
\sum_{s^{\prime}}\Pr\left(  s^{\prime}|s,a^{1},a^{2},a^{3}\right)
\mathfrak{u}^{n}\left(  s^{\prime}\right)  =\mathfrak{u}^{n}\left(
\mathbf{T}\left(  s,a^{n}\right)  \right)  .
\]
Hence, for $m=n$, (\ref{eq02015})\ becomes
\begin{equation}
\mathfrak{u}^{n}\left(  s\right)  =\max_{\pi^{n}\left(  s\right)  }\sum
_{a^{n}\in A^{n}\left(  s\right)  }\pi^{n}\left(  a^{n}|s\right)  \left[
q^{n}\left(  s\right)  +\gamma\mathfrak{u}^{n}\left(  \mathbf{T}\left(
s,a^{n}\right)  \right)  \right]  . \label{eq02016}%
\end{equation}
Furthermore let us define $\sigma_{\ast}^{n}\left(  s\right)  $ (for the
specific $s$ and $n$) by
\begin{equation}
\sigma_{\ast}^{n}\left(  s\right)  =\arg\max_{a^{n}\in A^{n}\left(  s\right)
}\left[  q^{n}\left(  s\right)  +\gamma\mathfrak{u}^{n}\left(  \mathbf{T}%
\left(  s,a^{n}\right)  \right)  \right]  . \label{eq02017}%
\end{equation}
If more than one state satisfy (\ref{eq02017}), we set $\sigma_{\ast}%
^{n}\left(  s\right)  $ to one of these states arbitrarily. Then, to maximize
the sum in (\ref{eq02016}) the $n$-th player must set $\pi^{n}\left(
\sigma_{\ast}^{n}\left(  s\right)  |s\right)  =1$ and $\pi^{n}\left(
a^{n}|s\right)  =0$ for all $a^{n}\neq\sigma_{\ast}^{n}\left(  s\right)  $.
Since this is true for all states and players (i.e., every player can, without
loss, use deterministic strategies) we also have $\mathfrak{u}^{n}\left(
s\right)  =u^{n}\left(  s\right)  $. Hence (\ref{eq02016}) becomes
\begin{equation}
u^{n}\left(  s\right)  =\max_{a^{n}\in A^{n}\left(  s\right)  }\left[
q^{n}\left(  s\right)  +\gamma u^{n}\left(  \mathbf{T}\left(  s,a^{n}\right)
\right)  \right]  =q^{n}\left(  s\right)  +\gamma u^{n}\left(  \mathbf{T}%
\left(  s,\sigma_{\ast}^{n}\left(  s\right)  \right)  \right)  .
\label{eq02018}%
\end{equation}
For $m\neq n$, the $m$-th player has no choice of action (i.e., $\sigma_{\ast
}^{m}\left(  s\right)  $ is the unique element of $A^{m}\left(  s\right)  $)
and (\ref{eq02016}) becomes
\begin{equation}
u^{m}\left(  s\right)  =q^{m}\left(  s\right)  +\gamma u^{m}\left(
\mathbf{T}\left(  s,\sigma_{\ast}^{n}\left(  s\right)  \right)  \right)  .
\label{eq02019}%
\end{equation}
We recognize that (\ref{eq02017})-(\ref{eq02019}) are (\ref{eq02012}%
)-(\ref{eq02013}). Also, (\ref{eq02017}) defines $\sigma_{\ast}^{n}\left(
s\right)  $ for every $n$ and $s$ and so we have obtained the required
deterministic positional strategies $\sigma_{\ast}=\left(  \sigma_{\ast}%
^{1},\sigma_{\ast}^{2},\sigma_{\ast}^{3}\right)  $.
\end{proof}

Note that the initial state $s_{0}$ plays no special role in the system
(\ref{eq02012})-(\ref{eq02013}). In other words, using the notation $u\left(
s\right)  =\left(  u^{1}\left(  s\right)  ,u^{2}\left(  s\right)
,u^{3}\left(  s\right)  \right)  $ and $\mathbf{u}=\left(  u\left(  s\right)
\right)  _{s\in S}$ (with the $G$ dependence suppressed) we see that
$\mathbf{u}$ and $\sigma_{\ast}$ are the same for every starting position
$s_{0}$ (i.e., for every $\Gamma\left(  G|s_{0}\right)  $). Also note that,
because of the structure of the payoffs, if $\Gamma\left(  G|s_{0}\right)  $
at some time $t_{1}$ reaches state $s_{1}$, the \textquotedblleft
remainder\textquotedblright\ game which is played from $t_{1}$ onward is
equivalent (modulo a payoff rescaling) to $\Gamma\left(  G|s_{1}\right)  $.
From these observations follows that, if the players use $\sigma_{\ast}$ in
$\Gamma\left(  G|s_{0}\right)  $ and at some time $t_{1}>0$ the game reaches
$s_{1}$, then $\sigma_{\ast}$ is an positional NE for both $\Gamma\left(
G|s_{0}\right)  $ and $\Gamma\left(  G|s_{1}\right)  $; the payoffs to the
players are $u\left(  s_{0}\right)  $ in the former and $u\left(
s_{1}\right)  $ in the latter.

Let us also note that Theorem \ref{prp0401} in fact holds for any
$\varepsilon\in\left[  0,1\right]  $; we have confined attention to the case
$\varepsilon\in\left[  0,\frac{1}{2}\right]  $ to represent the intuition that
the capturing cop's reward should be at least as large as that of the
non-capturing one's.

\subsection{Existence of Non-positional NE\label{sec0402}}

Next we prove that $\Gamma_{3}\left(  G|s_{0}\right)  $ also has deterministic
NE which are not positional. To this end we follow an approach which has
previously been used for several other $N$-player games of perfect information
\cite{boros2009,chatterjee2003,thuijsman1997}, namely the use of \emph{threat
strategies}.

We start by introducing, for $n\in\left\{  1,2,3\right\}  $, the auxiliary
games $\Gamma_{3}^{n}\left(  G|s_{0}\right)  $; these are two-player,
zero-sum, perfect-information games with movement sequence, states, action
sets, capturing conditions etc. being the same as in $\Gamma_{3}\left(
G|s_{0}\right)  $. However, in $\Gamma_{3}^{n}\left(  G|s_{0}\right)  $ player
$P_{n}$ controls token $n\ $and has payoff $Q^{n}$; and player $P_{-n}$
controls tokens $\left\{  1,2,3\right\}  \backslash\left\{  n\right\}  $ and
has payoff $-Q^{n}$. More specifically, the following hold.

\begin{enumerate}
\item $\Gamma_{3}^{3}\left(  G|s_{0}\right)  $ (played on $G$ with initial
state $s_{0}$) is the game where $P_{3}$, controlling $R$, plays against
$P_{-3}$, controlling $C_{1}$ and $C_{2}$; $P_{-3}$ has reward (and $P_{3}$
has penalty) equal to%
\begin{align*}
\gamma^{T_{C}}  &  :\text{ when either }C_{1}\text{ or }C_{2}\text{ captures
}R,\\
0  &  :\text{ when }R\text{ is not captured.}%
\end{align*}
It is easily seen that $\Gamma_{3}^{3}\left(  G|s_{0}\right)  $ is the
two-cops, one-robber modified CR\ game.

\item $\Gamma_{3}^{1}\left(  G|s_{0}\right)  $ (played on $G$ with initial
state $s_{0}$) is the game in which $P_{1}$, controlling $C_{1}$, plays
against $P_{-1}$, controlling $R$ and a \textquotedblleft
robber-friendly\textquotedblright\ $C_{2}$; $P_{1}$ receives reward (and
$P_{-1}$ receives penalty) equal to%
\begin{align*}
\left(  1-\varepsilon\right)  \gamma^{T_{C}}  &  :\text{when }C_{1}\text{
captures }R,\\
\varepsilon\gamma^{T_{C}}  &  :\text{when }C_{2}\text{ captures }R,\\
0  &  :\text{when }R\text{ is not captured.}%
\end{align*}

\item $\Gamma_{3}^{2}\left(  G|s_{0}\right)  $ is similar to $\Gamma_{3}%
^{1}\left(  G|s_{0}\right)  $, with the roles of $C_{1}$ and $C_{2}$ interchanged.
\end{enumerate}

\noindent It can be seen that in $\Gamma_{3}^{1}\left(  G|s\right)  $ an
optimal action plan for $P_{-1}$ is

\begin{enumerate}
\item when $c(G)=1$: $C_{2}$ and $R$ meet in the longest possible time but
before $R$ is caught by $C_{1}$;

\item when $c\left(  G\right)  >1$ and $C_{1}$ cannot alone capture $R$ (when
the game starts at $s_{0}$): $C_{2}$ always avoids $R$ and $R$ always avoids
both $C_{1}$ and $C_{2}$.
\end{enumerate}

Using the terminology and results of \cite{filar1996} we see that: for every
$n\in\left\{  1,2,3\right\}  $ and $s_{0}\in S$, $\Gamma_{3}^{n}%
(G|s_{0},\gamma,\varepsilon)$ is a two-player, zero-sum discounted stochastic
game with perfect information. Hence standard results \cite[Section
4.3]{filar1996} give the following.

\begin{lemma}
\label{prp0403}For each $n\in\left\{  1,2,3\right\}  $, $s_{0}\in S$ and
$\left(  \gamma,\varepsilon\right)  \in\Omega^{3}$, the game $\Gamma_{3}%
^{n}(G|s_{0},\gamma,\varepsilon)$ has a value and the players have optimal
deterministic positional strategies.
\end{lemma}

Returning to $\Gamma_{3}\left(  G|s\right)  $, the threat strategies are as
follows. The $n$-th player plays the strategy which is optimal for $P_{n}$ in
$\Gamma_{3}^{n}\left(  G|s\right)  $, as long as the other players do the
same. If at some point player $m$ deviates\footnote{We say that a player
\textquotedblleft deviates\textquotedblright\ from a strategy if he plays a
move different from the one prescribed by this strategy; since the game has
perfect information, this deviation will be immediately detected by the other
players.} from the above, then the $n$-th player (with $n\in\left\{
1,2,3\right\}  \backslash\left\{  m\right\}  $)\ adopts the strategy which
$P_{-m}$ uses for the $n$-th token in $\Gamma_{3}^{m}\left(  G|s\right)  $. In
other words, the threat strategy $\overline{\sigma}^{n}$ for the $n$-th player
is \textquotedblleft composed\textquotedblright\ by strategies $\phi_{m}^{n}$
as follows:
\begin{equation}
\overline{\sigma}^{n}=\left\{
\begin{array}
[c]{ll}%
\phi_{n}^{n} & \text{as long as every player }m\in\left\{  1,2,3\right\}
\backslash n\text{ follows }\phi_{m}^{m}\text{; }\\
\phi_{m}^{n} & \text{as soon as some player }m\in\left\{  1,2,3\right\}
\backslash n\ \text{ \textquotedblleft deviates\textquotedblright\ from }%
\phi_{m}^{m}\text{,}%
\end{array}
\right.  \label{eqthrstrat01}%
\end{equation}
where, in the game $\Gamma_{3}^{n}(G|s_{0},\gamma,\varepsilon)$ :

\begin{enumerate}
\item $\phi_{n}^{n}$ is an optimal strategy of $P_{n}$ against $P_{-n}$;\ 

\item $\phi_{n}^{m}$ (for $m\neq n$)\ is an optimal strategy used (for the
$m$-th token) by $P_{-n}$ against $P_{n}$.
\end{enumerate}

\noindent Since the $\phi_{n}^{m}$'s are positional they do not depend on the
starting state $s_{0}$; in fact the same $\phi_{n}^{m}$ is optimal for every
$s_{0}$ and corresponding game $\Gamma_{3}^{n}\left(  G|s_{0}\right)  $. We
now show that, for every $s_{0}$, $\overline{\sigma}=(\overline{\sigma}%
^{1},\overline{\sigma}^{2},\overline{\sigma}^{3})$ is a NE of $\Gamma
_{3}\left(  G|s_{0}\right)  $.

\begin{theorem}
\label{prp0404}For every graph $G$, $\left(  \gamma,\varepsilon\right)
\in\Omega^{3}$ and $s_{0}\in S$ in the game $\Gamma_{3}\left(  G|s_{0}%
,\gamma,\varepsilon\right)  $ we have%
\begin{equation}
\forall n\in\left\{  1,2,3\right\}  ,\forall\sigma^{n}:Q^{n}\left(
s_{0},\overline{\sigma}^{1},\overline{\sigma}^{2},\overline{\sigma}%
^{3}\right)  \geq Q^{n}(s_{0},\sigma^{n},\overline{\sigma}^{-n}) \label{eq001}%
\end{equation}
where $\overline{\sigma}^{n}$ (for $n\in\left\{  1,2,...,N\right\}  $) is a
strategy of the form defined in (\ref{eqthrstrat01}).
\end{theorem}

\begin{proof}
Recall that we can write payoffs in any of the equivalent forms:$\ Q^{n}%
\left(  s_{0},\sigma\right)  $, $Q^{n}\left(  s_{0},s_{1},s_{2},...\right)  $,
$Q^{n}\left(  h\right)  $ (where $h=\left(  s_{0},s_{1},s_{2,}....\right)  $).

We choose some initial state $s$ and fix it for the rest of the proof. Now let
us prove (\ref{eq001}) for the case $n=1$. In other words, we need to show
that%
\begin{equation}
\forall\sigma^{1}:Q^{1}(s,\overline{\sigma}^{1},\overline{\sigma}%
^{2},\overline{\sigma}^{3})\geq Q^{1}(s,\sigma^{1},\overline{\sigma}%
^{2},\overline{\sigma}^{3}).\label{eq002a}%
\end{equation}
We take any $\sigma^{1}$ and let%
\begin{align*}
\widehat{h} &  =(\widehat{s}_{0},\widehat{s}_{1},\widehat{s}_{2},...)\text{ be
the history produced by }(\overline{\sigma}^{1},\overline{\sigma}%
^{2},\overline{\sigma}^{3})\ \text{and initial state }\widehat{s}_{0}=s\text{,
}\\
\widetilde{h} &  =(\widetilde{s}_{0},\widetilde{s}_{1},\widetilde{s}%
_{2},...)\text{ be the history produced by }(\sigma^{1},\overline{\sigma}%
^{2},\overline{\sigma}^{3})\text{ and initial state }\widetilde{s}%
_{0}=s=\widehat{s}_{0}\text{.}%
\end{align*}
We also define $T_{1}$ as the earliest time in which $(\sigma^{1}%
,\overline{\sigma}^{2},\overline{\sigma}^{3})$ produce different states:
\[
T_{1}=\min\left\{  t:\widetilde{s}_{t}\neq\widehat{s}_{t}\right\}  .
\]
If $T_{1}=\infty$, then $\widetilde{h}=\widehat{h}$ and (\ref{eq002a}) holds
with equality:%
\begin{equation}
Q^{1}(s,\overline{\sigma}^{1},\overline{\sigma}^{2},\overline{\sigma}%
^{3})=Q^{1}\mathbf{(}\widehat{h})=Q^{1}\mathbf{(}\widetilde{h})=Q^{1}%
(\sigma^{1},\overline{\sigma}^{2},\overline{\sigma}^{3}).\label{eq002b}%
\end{equation}
If $T_{1}<\infty$, then at $t=T_{1}$ player 1 deviated from $\phi_{1}^{1}$,
the first difference in states appeared and it was detected by players 2 and
3, who switched to $\phi_{1}^{2}$ and $\phi_{1}^{3}$, respectively. We have%
\begin{align}
Q^{1}(s,\overline{\sigma}^{1},\overline{\sigma}^{2},\overline{\sigma}^{3}) &
=Q^{1}\mathbf{(}\widehat{h})=\sum_{t=0}^{T_{1}-2}\gamma^{t}q^{1}\left(
\widehat{s}_{t}\right)  +\sum_{t=T_{1}-1}^{\infty}\gamma^{t}q^{1}\left(
\widehat{s}_{t}\right)  ,\label{eq003a}\\
Q^{1}(s,\sigma^{1},\overline{\sigma}^{2},\overline{\sigma}^{3}) &
=Q^{1}\mathbf{(}\widetilde{h})=\sum_{t=0}^{T_{1}-2}\gamma^{t}q^{1}\left(
\widetilde{s}_{t}\right)  +\sum_{t=T_{1}-1}^{\infty}\gamma^{t}q^{1}\left(
\widetilde{s}_{t}\right)  .\label{eq004a}%
\end{align}
Since $\widetilde{s}_{t}=\widehat{s}_{t}$ for every $t<T_{1}$, it suffices to
compare the second sums of (\ref{eq003a}) and (\ref{eq004a}). In what follows
we let $s^{\ast}=\widehat{s}_{T_{1}-1}=\widetilde{s}_{T_{1}-1}$. 

\begin{enumerate}
\item Consider first $\widehat{h}=(\widehat{s}_{0},\widehat{s}_{1},\widehat
{s}_{2},...)$. It is produced by $\overline{\sigma}=(\overline{\sigma}%
^{1},\overline{\sigma}^{2},\overline{\sigma}^{3})$ which means that the entire
$\widehat{h}$ is actually produced by $\left(  \phi_{1}^{1},\phi_{2}^{2}%
,\phi_{3}^{3}\right)  $. Since  every $\phi_{m}^{m}\ $is positional, the
history $(\widehat{s}_{0},\widehat{s}_{1},...,\widehat{s}_{T_{1}-2})$ does not
influence the moves produced at times $T_{1},T_{1}+1,...$. Hence we have%
\begin{equation}
\sum_{t=T_{1}-1}^{\infty}\gamma^{t}q^{1}\left(  \widehat{s}_{t}\right)
=\gamma^{T_{1}-1}\sum_{t=0}^{\infty}\gamma^{t}q^{1}\left(  \widehat{s}%
_{T_{1}-1+t}\right)  =\gamma^{T_{1}-1}Q^{1}(s^{\ast},\phi_{1}^{1},\phi_{2}%
^{2},\phi_{3}^{3}).\label{eq005a}%
\end{equation}
In other words, the sum in (\ref{eq005a}) is proportional to the payoff of
player 1\ in $\Gamma_{3}\left(  G|s^{\ast}\right)  $ when each player
$n\in\{1,2,3\}$ uses strategy $\phi_{n}^{n}$. But $Q^{1}(s^{\ast},\phi_{1}%
^{1},\phi_{2}^{2},\phi_{3}^{3})$ is also the payoff of $P_{1}$ in $\Gamma
_{3}^{1}\left(  G|s^{\ast}\right)  $ (which starts at $s^{\ast}$) with $P_{1}$
playing $\phi_{1}^{1}$ and $P_{-1}$ playing $(\phi_{2}^{2},\phi_{3}^{3})$.
However, in $\Gamma_{3}^{1}\left(  G|s^{\ast}\right)  $ the optimal strategy
of $P_{-1}$ is $\left(  \phi_{1}^{2},\phi_{1}^{3}\right)  $; hence we have the
following%
\begin{equation}
\gamma^{T_{1}-1}Q^{1}(s^{\ast},\phi_{1}^{1},\phi_{2}^{2},\phi_{3}^{3}%
)\geq\gamma^{T_{1}-1}Q^{1}(s^{\ast},\phi_{1}^{1},\phi_{1}^{2},\phi_{1}%
^{3}).\label{eq005b}%
\end{equation}

\item Next consider $\widetilde{h}=(\widetilde{s}_{0},\widetilde{s}%
_{1},\widetilde{s}_{2},...)$. It is produced by $(\sigma^{1},\overline{\sigma
}^{2},\overline{\sigma}^{3})$ and, since $\sigma^{1}$ is not necessarily
positional, $\widetilde{s}_{T_{1}},\widetilde{s}_{T_{1}+1},\widetilde
{s}_{T_{1}+2}...$ \ could depend on $(\widetilde{s}_{0},\widetilde{s}%
_{1},...,\widetilde{s}_{T_{1}-2})$. However, we can introduce the strategy
$\rho^{1}$ \emph{induced} by $\sigma^{1}$ on the game starting at $s^{\ast}$,
which will produce the same history $(\widetilde{s}_{T_{1}},\widetilde
{s}_{T_{1}+1},\widetilde{s}_{T_{1}+2},...)$ as $\sigma^{1}$.\footnote{We
define $\rho^{1}$ such that, when combined with $\widetilde{s}_{T_{1}-1}%
,\phi_{1}^{2},\phi_{1}^{3}$, will produce the same history $(\widetilde
{s}_{T_{1}},\widetilde{s}_{T_{1}+1},\widetilde{s}_{T_{1}+2},...)$ as
$\sigma^{1}$. Note that $\rho^{1}$ will in general depend (in an indirect way)
on $(\widetilde{s}_{0},\widetilde{s}_{1},...,\widetilde{s}_{T_{1}-2})$.} Then,
from the optimality of $\phi_{1}^{1}$ as a response to $\left(  \phi_{1}%
^{2},\phi_{1}^{3}\right)  $ in $\Gamma_{3}^{1}\left(  G|s^{\ast}\right)  $, we
have%
\begin{equation}
\sum_{t=T_{1}-1}^{\infty}q^{1}\left(  \widetilde{s}_{t}\right)  =\gamma
^{T_{1}-1}Q^{1}(s^{\ast},\rho^{1},\phi_{1}^{2},\phi_{1}^{3})\leq\gamma
^{T_{1}-1}Q^{1}(s^{\ast},\phi_{1}^{1},\phi_{1}^{2},\phi_{1}^{3}%
).\label{eq006a}%
\end{equation}

\end{enumerate}

\noindent Combining (\ref{eq003a})-(\ref{eq006a}) we have:%
\begin{align*}
Q^{1}(s,\sigma^{1},\overline{\sigma}^{2},\overline{\sigma}^{3})  &
=\sum_{t=0}^{T_{1}-2}\gamma^{t}q^{1}\left(  \widetilde{s}_{t}\right)
+\gamma^{T_{1}-1}Q^{1}(s^{\ast},\rho^{1},\phi_{1}^{2},\phi_{1}^{3})\leq
\sum_{t=0}^{T_{1}-2}\gamma^{t}q^{1}\left(  \widetilde{s}_{t}\right)
+\gamma^{T_{1}-1}Q^{1}(s^{\ast},\phi_{1}^{1},\phi_{1}^{2},\phi_{1}^{3})\\
&  \leq\sum_{t=0}^{T_{1}-2}\gamma^{t}q^{1}\left(  \widehat{s}_{t}\right)
+\gamma^{T_{1}-1}Q^{1}(s^{\ast},\phi_{1}^{1},\phi_{2}^{2},\phi_{3}^{3}%
)=Q^{1}(s,\overline{\sigma}^{1},\overline{\sigma}^{2},\overline{\sigma}^{3})
\end{align*}
and we have proved (\ref{eq002a}), which is (\ref{eq001})\ for $n=1$. The
proof for the cases $n=2$ and $n=3$ are similar and hence omitted.
\end{proof}

The strategies $\left(  \overline{\sigma}^{1},\overline{\sigma}^{2}%
,\overline{\sigma}^{3}\right)  $ are \emph{not}\ positional. In particular,
the action of a player at time $t$ may be influenced by the action (deviation)
performed by another player at time $t-2$. Hence, $\left(  \overline{\sigma
}^{1},\overline{\sigma}^{2},\overline{\sigma}^{3}\right)  $ is a
\emph{non}-positional NE.

Just like Theorem \ref{prp0401}, Theorem \ref{prp0404} actually holds for any
$\varepsilon\in\mathbb{[}0,1]$.

\subsection{Cop Number, Capturing and Non-capturing NE\label{sec0403}}

In this section we examine the connection of $c\left(  G\right)  $ to the
existence of capturing and non-capturing NE in $\Gamma_{3}\left(
G|s_{0},\gamma,\varepsilon\right)  $.

\begin{theorem}
\label{prp0405}For any $G$ with $c\left(  G\right)  =1$ the following holds:%
\[
\forall\left(  \gamma,\varepsilon\right)  \in\Omega^{3},\forall s_{0}\in
S:\text{every NE of }\Gamma_{3}\left(  G|s_{0},\gamma,\varepsilon\right)
\text{ is capturing. }%
\]

\end{theorem}

\begin{proof}
Suppose $c\left(  G\right)  =1$, take any $s_{0}\in S_{NC}$ (the case
$s_{0}\in S_{C}$ is trivial) and let $\left(  \sigma^{1},\sigma^{2},\sigma
^{3}\right)  $ be a NE of $\Gamma\left(  G|s_{0}\right)  $. Suppose it is a
non-capturing NE; then we have%
\begin{equation}
\forall\rho^{1}:0=Q^{1}(s_{0},\sigma^{1},\sigma^{2},\sigma^{3})\geq
Q^{1}(s_{0},\rho^{1},\sigma^{2},\sigma^{3})\text{.} \label{eq030301a}%
\end{equation}
Now take $\widehat{\sigma}^{1}$ to be an optimal cop strategy \emph{in CR with
one cop}\footnote{More precisely, define $\widehat{\sigma}^{1}$ by%
\[
\forall x^{1},x^{2},x^{3},n:\widehat{\sigma}^{1}\left(  x^{1},x^{2}%
,x^{3},n\right)  =\widetilde{\sigma}^{1}\left(  x^{1},x^{3},n\right)
\]
where $\widetilde{\sigma}^{1}$ is an optimal cop strategy in one-cop CR. We
will repeatedly use, without further comment,\ this method to produce SCAR
strategies from strategies in CR.}. In this game, since $c\left(  G\right)
=1$, $C_{1}$ will capture $R$ in some finite time, which will depend on $R$'s
strategy but will be bounded above by the $T_{2}\left(  G\right)  $ defined in
Section \ref{sec03}. When $C_{1}$ uses $\widehat{\sigma}^{1}$ in $\Gamma
_{3}\left(  G|s_{0}\right)  $, $C_{2}$ may influence the game by capturing $R$
no later than $C_{1}$. Hence we have the following possibilities.

\begin{enumerate}
\item $C_{1}$ captures $R$ at some time $T_{1}$.

\item $C_{2}$ captures $R$ before $C_{1}$, i.e., at some time $T_{2}<T_{1}$.

\item $C_{1}$ and $C_{2}$ capture $R$ simultaneously at some time $T_{12}$.
\end{enumerate}

\noindent At any rate, we will have $\max\left(  T_{1},T_{2},T_{12}\right)
$\noindent$\leq T_{2}\left(  G\right)  <\infty$. Hence $C_{1}$ will receive
payoff $Q^{1}(s_{0},\widehat{\sigma}^{1},\sigma^{2},\sigma^{3})$ which will
satisfy
\begin{equation}
Q^{1}(s_{0},\widehat{\sigma}^{1},\sigma^{2},\sigma^{3})\geq\min\left(
\gamma^{T_{1}}\left(  1-\varepsilon\right)  ,\gamma^{T_{2}}\varepsilon
,\gamma^{T_{12}}\frac{1}{2}\right)  >0. \label{eq030302a}%
\end{equation}
But (\ref{eq030302a}) contradicts (\ref{eq030301a}); hence $\left(  \sigma
^{1},\sigma^{2},\sigma^{3}\right)  $ must be a \emph{capturing} NE of
$\Gamma_{3}\left(  G|s_{0}\right)  $. We conclude that every NE of $\Gamma
_{3}\left(  G|s_{0}\right)  $ is capturing.
\end{proof}

It might be assumed that $\Gamma_{3}\left(  G\right)  $ is equivalent to two
one-cop CR\ games played on the same graph and, consequently, each cop should
use an optimal CR\ strategy. For example, it might be assumed that on graphs
$G$ with $c(G)=1$, if $\widehat{\sigma}^{1}$, $\widehat{\sigma}^{2}$,
$\widehat{\sigma}^{3}$ are time optimal cop and robber strategies in one-cop
CR, then $\widehat{\sigma}=\left(  \widehat{\sigma}^{1},\widehat{\sigma}%
^{2},\widehat{\sigma}^{3}\right)  $ is a NE of $\Gamma_{3}\left(  G\right)  $.
This is not true; in certain cases a cop may want to delay capture to ensure
that it is effected by him, as seen in the following example.

\begin{example}
\label{prp0406}\normalfont Suppose $\Gamma_{3}\left(  G\right)  $ is played on
the graph $G$ of Fig. \ref{fig0301} with the initial positions indicated;
$C_{1}$ has the first move. \ Further, take $\varepsilon<\frac{1}{2}%
$.\begin{figure}[ptbh]
\begin{center}
\begin{tikzpicture}
\SetGraphUnit{2}
\Vertex[x= 1,y= 0]{1}
\Vertex[x= 2,y= 0]{2}
\Vertex[x= 3,y= 0]{3}
\Vertex[x= 4,y= 0]{4}
\Vertex[x= 5,y= 0]{5}
\Vertex[x= 6,y= 0]{6}
\Vertex[x= 7,y= 0]{7}
\Vertex[x= 5,y= 1]{8}
\Vertex[x= 5,y= 2]{9}
\node(A) [label=$C_1$] at (6,-1) {};
\node(B) [label=$C_2$] at (1,-1) {};
\node(C) [label=$R  $] at (4,-1) {};
\Edge(1)(2)
\Edge(2)(3)
\Edge(3)(4)
\Edge(4)(5)
\Edge(5)(6)
\Edge(6)(7)
\Edge(5)(8)
\Edge(8)(9)
\SetVertexNoLabel
\end{tikzpicture}
\end{center}
\par
\label{fig0301}\caption{An example where minimizing capture time does not
yield a NE.}%
\end{figure}
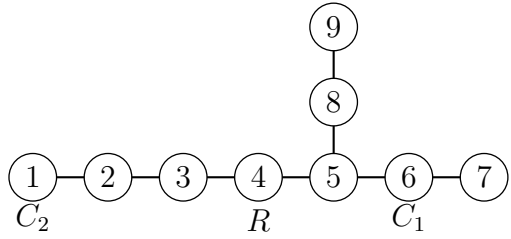\ Let $\widehat{\sigma}^{n}$ ($n\in\left\{  1,2\right\}  $) be an
(one-cop)\ CR\ optimal strategy of the $n$-th cop; in this case it consists in
each cop moving towards the robber at every turn. Now suppose that for these
$\left(  \widehat{\sigma}^{1},\widehat{\sigma}^{2}\right)  $ there exists a NE
$\left(  \widehat{\sigma}^{1},\widehat{\sigma}^{2},\widehat{\sigma}%
^{3}\right)  $; it is easily seen that $\widehat{\sigma}^{3}$ must be an
optimal robber strategy in two-cop CR. If the players used $\left(
\widehat{\sigma}^{1},\widehat{\sigma}^{2},\widehat{\sigma}^{3}\right)  $ the
game would evolve as follows:%
\[%
\begin{tabular}
[c]{|l|l|l|l|l|l|l|}\hline
\textbf{Turn} & 0 & 1 & 2 & 3 & 4 & 5\\\hline
$C_{1}$\textbf{ vertex} & 6 & 5 & 5 & 5 & 4 & 4\\\hline
$C_{2}$\textbf{ vertex} & 1 & 1 & 2 & 2 & 2 & 3\\\hline
$R$\textbf{ vertex} & 4 & 4 & 4 & 3 & 3 & 3\\\hline
\end{tabular}
\ \ \ \ \ \ \ \
\]
Note that the robber will move so as to be captured by $C_{2}$, because this
increases capture time by 1. So the payoffs are
\begin{align*}
Q^{1}\left(  s_{0},\widehat{\sigma}^{1},\widehat{\sigma}^{2},\widehat{\sigma
}^{3}\right)   &  =\gamma^{5}\varepsilon,\\
Q^{2}\left(  s_{0},\widehat{\sigma}^{1},\widehat{\sigma}^{2},\widehat{\sigma
}^{3}\right)   &  =\gamma^{5}\left(  1-\varepsilon\right)  ,\\
Q^{3}\left(  s_{0},\widehat{\sigma}^{1},\widehat{\sigma}^{2},\widehat{\sigma
}^{3}\right)   &  =-\gamma^{5}.
\end{align*}
Now, suppose $C_{2},R$ stick to their strategies, while $C_{1}$ uses the
following strategy $\widetilde{\sigma}^{1}$: on his first move he retreats to
vertex 7 and afterwards moves directly towards the robber. This frees a path
for $R$ towards 9, which increases the capture time. The game evolves as
follows.%
\[%
\begin{tabular}
[c]{|l|l|l|l|l|l|l|l|l|l|l|l|l|l|l|}\hline
\textbf{Turn} & 0 & 1 & 2 & 3 & 4 & 5 & 6 & 7 & 8 & 9 & 10 & 11 & 12 &
13\\\hline
$C_{1}$\textbf{ vertex} & 6 & 7 & 7 & 7 & 6 & 6 & 6 & 5 & 5 & 5 & 8 & 8 & 8 &
9\\\hline
$C_{2}$\textbf{ vertex} & 1 & 1 & 2 & 2 & 2 & 3 & 3 & 3 & 4 & 4 & 4 & 5 & 5 &
5\\\hline
$R$\textbf{ vertex} & 4 & 4 & 4 & 5 & 5 & 5 & 8 & 8 & 8 & 9 & 9 & 9 & 9 &
9\\\hline
\end{tabular}
\ \ \ \ \ \ \ \
\]
So the payoffs are%
\begin{align*}
Q^{1}\left(  s_{0},\widetilde{\sigma}^{1},\widehat{\sigma}^{2},\widehat
{\sigma}^{3}\right)   &  =\gamma^{13}\left(  1-\varepsilon\right)  ,\\
Q^{2}\left(  s_{0},\widetilde{\sigma}^{1},\widehat{\sigma}^{2},\widehat
{\sigma}^{3}\right)   &  =\gamma^{13}\varepsilon,\\
Q^{3}\left(  s_{0},\widetilde{\sigma}^{1},\widehat{\sigma}^{2},\widehat
{\sigma}^{3}\right)   &  =-\gamma^{13}.
\end{align*}
It is easy to see that%
\[
\gamma>\left(  \frac{\varepsilon}{1-\varepsilon}\right)  ^{1/8}\Rightarrow
\left(  1-\varepsilon\right)  \gamma^{13}>\varepsilon\gamma^{5}\Rightarrow
Q^{1}\left(  s_{0},\widetilde{\sigma}^{1},\widehat{\sigma}^{2},\widehat
{\sigma}^{3}\right)  >Q^{1}\left(  s_{0},\widehat{\sigma}^{1},\widehat{\sigma
}^{2},\widehat{\sigma}^{3}\right)  .
\]
Since $C_{1}$ can unilaterally improve his payoff, $\left(  \widehat{\sigma
}^{1},\widehat{\sigma}^{2},\widehat{\sigma}^{3}\right)  $ cannot be a NE\ of
$\Gamma_{3}\left(  G|s_{0},\gamma,\varepsilon\right)  $.
\end{example}

\FloatBarrier

We now move to graphs with cop number greater than one.

\begin{theorem}
\label{prp0407}For any $G$ with $c\left(  G\right)  =2$ the following holds:%
\[
\forall\left(  \gamma,\varepsilon\right)  \in\Omega^{3},\forall s_{0}\in
S:\text{there exists a capturing NE of }\Gamma_{3}\left(  G|s_{0}%
,\gamma,\varepsilon\right)  .
\]

\end{theorem}

\begin{proof}
Take any $G$ with $c\left(  G\right)  =2$, any $\left(  \gamma,\varepsilon
\right)  \in\Omega^{3}$ and any $s_{0}\in S_{NC}$ (the case $s_{0}\in S_{C}$
is trivial) and fix them for the rest of the proof. Now take any threat
strategy profile $\overline{\sigma}=\left(  \overline{\sigma}^{1}%
,\overline{\sigma}^{2},\overline{\sigma}^{3}\right)  $; according to Theorem
\ref{prp0404}, $\overline{\sigma}$ is a NE and it can be either capturing or
non-capturing. If it is capturing we are done; let us then suppose that
$\overline{\sigma}$ is non-capturing. Recall that, for all $n\in\left\{
1,2,3\right\}  $:
\[
\overline{\sigma}^{n}=\left\{
\begin{array}
[c]{ll}%
\phi_{n}^{n} & \text{as long as every player }m\in\left\{  1,2,3\right\}
\backslash n\text{ follows }\phi_{m}^{m}\text{; }\\
\phi_{m}^{n} & \text{as soon as some player }m\in\left\{  1,2,3\right\}
\backslash n\ \text{ deviates from }\phi_{m}^{m}\text{,}%
\end{array}
\right.
\]
where in the game $\Gamma_{3}^{n}\left(  G|s_{0},\gamma,\varepsilon\right)  $:

\begin{enumerate}
\item $\phi_{n}^{n}$ is an optimal strategy of $P_{n}$ against $P_{-n}$;\ 

\item $\phi_{n}^{m}$ (for $m\neq n$) is an optimal strategy used (for the
$m$-th token) by $P_{-n}$ against $P_{n}$.
\end{enumerate}

As mentioned, when $\overline{\sigma}$ is used in $\Gamma_{3}\left(
G|s_{0},\gamma,\varepsilon\right)  $ the $n$-th player (for $n\in\left\{
1,2,3\right\}  $)\ will follow strategy $\phi_{n}^{n}$ for the entire game. We
will now construct a new profile $\widetilde{\sigma}=\left(  \widetilde
{\sigma}^{1},\widetilde{\sigma}^{2},\widetilde{\sigma}^{3}\right)  $\ which
\emph{will} be a capturing NE of $\Gamma_{3}\left(  G|s_{0},\gamma
,\varepsilon\right)  $. To this end we first select an optimal strategy
profile $\widehat{\sigma}=\left(  \widehat{\sigma}^{1},\widehat{\sigma}%
^{2},\widehat{\sigma}^{3}\right)  $\ in the two-cops CR; since $c\left(
G\right)  =2$, $\widehat{\sigma}$ will be capturing for every $s_{0}$. Then,
for each $n\in\left\{  1,2,3\right\}  $, we let
\[
\widetilde{\sigma}^{n}=\left\{
\begin{array}
[c]{ll}%
\widehat{\sigma}^{n} & \text{as long as every player }m\in\left\{
1,2,3\right\}  \backslash n\text{ follows }\widehat{\sigma}^{m}\text{; }\\
\phi_{m}^{n} & \text{as soon as some player }m\in\left\{  1,2,3\right\}
\backslash n\ \text{ deviates from }\widehat{\sigma}^{m}\text{.}%
\end{array}
\right.
\]
The $\phi_{m}^{n}$'s above are the same as in $\overline{\sigma}$. Hence, to
show that $\widetilde{\sigma}$\ is \ of the form prescribed by Theorem
\ref{prp0404}, we have to show that each $\widehat{\sigma}^{n}$ is optimal in
the corresponding $\Gamma_{3}^{n}\left(  G|s_{0},\gamma,\varepsilon\right)  $.
In the following arguments we will repeatedly use the (easily verified) fact
that $\Gamma_{3}\left(  G|s_{0},\gamma,\varepsilon\right)  $ is
path-equivalent to $\Gamma_{3}^{n}\left(  G|s_{0},\gamma,\varepsilon\right)  $
for every $n\in\left\{  1,2,3\right\}  $.

\begin{enumerate}
\item First consider $C_{1}$ playing $\phi_{1}^{1}$ in $\Gamma_{3}\left(
G|s_{0},\gamma,\varepsilon\right)  $; since by assumption $\overline{\sigma}$
is non-capturing in $\Gamma_{3}\left(  G|s_{0},\gamma,\varepsilon\right)  $,
it is also non-capturing in $\Gamma_{3}^{1}\left(  G|s_{0},\gamma
,\varepsilon\right)  $. Hence $C_{1}$ playing $\phi_{1}^{1}$ against $\phi
_{2}^{2}$ and $\phi_{3}^{3}$ will receive a payoff of zero, in both
$\Gamma_{3}\left(  G|s_{0},\gamma,\varepsilon\right)  $ and $\Gamma_{3}%
^{1}\left(  G|s_{0},\gamma,\varepsilon\right)  $. But then $C_{1}$ playing
$\phi_{1}^{1}$ against $\phi_{1}^{2}$ and $\phi_{1}^{3}$ (which are optimal in
$\Gamma_{3}^{1}\left(  G|s_{0},\gamma,\varepsilon\right)  $)\ will also
receive a payoff of zero. It follows that $C_{1}$'s optimal payoff in
$\Gamma_{3}^{1}\left(  G|s_{0},\gamma,\varepsilon\right)  $ is zero and hence
\emph{any} strategy is optimal for him in $\Gamma_{3}^{1}\left(
G|s_{0},\gamma,\varepsilon\right)  $ and so is, in particular, $\widehat
{\sigma}^{1}$.

\item By a similar argument, any strategy, and in particular $\widehat{\sigma
}^{2}$, will be optimal for $C_{2}$ in $\Gamma_{3}^{2}\left(  G|s_{0}%
,\gamma,\varepsilon\right)  $.

\item Finally, the game $\Gamma_{3}^{3}\left(  G|s_{0},\gamma,\varepsilon
\right)  $ is the two-cop CR, and hence $\widehat{\sigma}^{3}$ will be optimal
for $R$ in $\Gamma_{3}^{3}\left(  G|s_{0},\gamma,\varepsilon\right)  $.
\end{enumerate}

Using the above observations we see that, according to Theorem \ref{prp0404},
$\widetilde{\sigma}$\ is a NE (in threat strategies) of $\Gamma_{3}\left(
G|s_{0},\gamma,\varepsilon\right)  $. Furthermore, playing $\widetilde{\sigma
}$\ at equilibrium is equivalent to playing $\widehat{\sigma}$, a capturing
profile in both $\Gamma_{3}^{3}\left(  G|s_{0},\gamma,\varepsilon\right)  $
(i.e., two-cop CR)\ and $\Gamma_{3}\left(  G|s_{0},\gamma,\varepsilon\right)
$. Hence $\widetilde{\sigma}$\ \ is a capturing NE of $\Gamma_{3}\left(
G|s_{0},\gamma,\varepsilon\right)  $.
\end{proof}

\begin{remark}
\label{prp0408}\normalfont Note that the NE\ of the above theorem is not positional.
\end{remark}

The next theorem holds on a restricted set of $\left(  \gamma,\varepsilon
\right)  $ values:
\[
\widetilde{\Omega}^{3}=\left\{  \left(  \gamma,\varepsilon\right)  :\gamma
\in\left(  0,1\right)  ,\varepsilon\in\left[  0,\frac{1}{2}\right]
,\gamma<\frac{\varepsilon}{1-\varepsilon}\right\}  .
\]

\begin{theorem}
\label{prp0409}For any $G$ with $c\left(  G\right)  =2$, let $\widehat{\sigma
}=\left(  \widehat{\sigma}^{1},\widehat{\sigma}^{2},\widehat{\sigma}%
^{3}\right)  $ be an optimal strategy profile in the two-cop CR game. Then the
following holds:%
\[
\forall\left(  \gamma,\varepsilon\right)  \in\widetilde{\Omega}^{3},\forall
s_{0}\in S:\widehat{\sigma}\text{ is a capturing NE of }\Gamma_{3}\left(
G|s_{0},\gamma,\varepsilon\right)  .
\]

\end{theorem}

\begin{proof}
Take any $G$ with $c\left(  G\right)  =2$, any $\left(  \gamma,\varepsilon
\right)  \in\widetilde{\Omega}^{3}$ and any $\left(  \widehat{\sigma}%
^{1},\widehat{\sigma}^{2},\widehat{\sigma}^{3}\right)  $ which is optimal in
the two-cop CR game; we fix these for the rest of the proof. Obviously
$\widehat{\sigma}$ is a capturing profile, since it is optimal in CR and
$c\left(  G\right)  =2$. So we need to show that it is also a NE of
$\Gamma_{3}\left(  G|s_{0},\gamma,\varepsilon\right)  $. This will obviously
be true when $s_{0}\in S_{C}$, so let us consider any $s_{0}\in S_{NC}$. Let
$T_{1}$ be the capture time corresponding to $\left(  s_{0},\widehat{\sigma
}^{1},\widehat{\sigma}^{2},\widehat{\sigma}^{3}\right)  $. This is the same in
both CR and $\Gamma_{3}\left(  G|s_{0},\gamma,\varepsilon\right)  $, since the
two games are path-equivalent.

Assume for the time being that the capturing cop is $C_{1}$; then the payoffs
are%
\begin{align*}
Q^{1}\left(  s_{0},\widehat{\sigma}^{1},\widehat{\sigma}^{2},\widehat{\sigma
}^{3}\right)   &  =\left(  1-\varepsilon\right)  \gamma^{T_{1}},\\
Q^{2}\left(  s_{0},\widehat{\sigma}^{1},\widehat{\sigma}^{2},\widehat{\sigma
}^{3}\right)   &  =\varepsilon\gamma^{T_{1}},\\
Q^{3}\left(  s_{0},\widehat{\sigma}^{1},\widehat{\sigma}^{2},\widehat{\sigma
}^{3}\right)   &  =-\gamma^{T_{1}}.
\end{align*}
We will show that no player can improve his payoff by unilaterally changing
his strategy.

\begin{enumerate}
\item Suppose $R$ uses some strategy $\sigma^{3}$ and the capture time of
$\left(  s_{0},\widehat{\sigma}^{1},\widehat{\sigma}^{2},\sigma^{3}\right)  $
is $T_{2}$. By the optimality (in CR) of $\widehat{\sigma}^{1},\widehat
{\sigma}^{2},\widehat{\sigma}^{3}$, we have $T_{2}\leq T_{1}$ and so
\[
Q^{3}\left(  s_{0},\widehat{\sigma}^{1},\widehat{\sigma}^{2},\sigma
^{3}\right)  =-\gamma^{T_{2}}\leq-\gamma^{T_{1}}=Q^{3}\left(  s_{0}%
,\widehat{\sigma}^{1},\widehat{\sigma}^{2},\widehat{\sigma}^{3}\right)  .
\]
So $R$ has no motive to deviate from $\widehat{\sigma}^{3}$.

\item Similarly, suppose $C_{1}$ uses some strategy $\sigma^{1}$ and the
capture time of $\left(  s_{0},\sigma^{1},\widehat{\sigma}^{2},\widehat
{\sigma}^{3}\right)  $ is $T_{2}$; if $T_{2}=\infty$ we have no capture;
otherwise capture can be effected by either $C_{1}$ or $C_{2}$. At any rate,
by the optimality of $\left(  \widehat{\sigma}^{1},\widehat{\sigma}%
^{2},\widehat{\sigma}^{3}\right)  $, we have $T_{2}\geq T_{1}$ and the maximum
possible payoff to $C_{1}$ is $\left(  1-\varepsilon\right)  \gamma^{T_{2}}$.
Since
\[
Q^{1}\left(  s_{0},\sigma^{1},\widehat{\sigma}^{2},\widehat{\sigma}%
^{3}\right)  \leq\left(  1-\varepsilon\right)  \gamma^{T_{2}}\leq\left(
1-\varepsilon\right)  \gamma^{T_{1}}=Q^{1}\left(  s_{0},\widehat{\sigma}%
^{1},\widehat{\sigma}^{2},\widehat{\sigma}^{3}\right)  ,
\]
$C_{1}$ has no motive to deviate from $\widehat{\sigma}^{1}$.

\item Finally, suppose $C_{2}$ uses some strategy $\sigma^{2}$ and the capture
time of $\left(  s_{0},\widehat{\sigma}^{1},\sigma^{2},\widehat{\sigma}%
^{3}\right)  $ is $T_{2}$. If $T_{2}=\infty$ we have no capture; otherwise
capture can be effected by either $C_{1}$ or $C_{2}$. If we have no capture
then%
\[
Q^{2}\left(  s_{0},\widehat{\sigma}^{1},\sigma^{2},\widehat{\sigma}%
^{3}\right)  =0<\varepsilon\gamma^{T_{1}}=Q^{2}\left(  s_{0},\widehat{\sigma
}^{1},\widehat{\sigma}^{2},\widehat{\sigma}^{3}\right)  .
\]
If capture is effected by $C_{1}$, we have $T_{2}\geq T_{1}$ and
\[
Q^{2}\left(  s_{0},\widehat{\sigma}^{1},\sigma^{2},\widehat{\sigma}%
^{3}\right)  =\varepsilon\gamma^{T_{2}}\leq\varepsilon\gamma^{T1}=Q^{2}\left(
s_{0},\widehat{\sigma}^{1},\widehat{\sigma}^{2},\widehat{\sigma}^{3}\right)
.
\]
Finally, if capture is effected by $C_{2}$, we have $T_{2}\geq T_{1}+1$ (if
$C_{2}$ could capture before $C_{1}$ this would be achieved by $\left(
s_{0},\widehat{\sigma}^{1},\widehat{\sigma}^{2},\widehat{\sigma}^{3}\right)
$) and, since $\left(  \gamma,\varepsilon\right)  \in\widetilde{\Omega^{3}}$
implies $\gamma<\frac{\varepsilon}{1-\varepsilon}$, we have%
\[
Q^{2}\left(  s_{0},\widehat{\sigma}^{1},\sigma^{2},\widehat{\sigma}%
^{3}\right)  =\left(  1-\varepsilon\right)  \gamma^{T_{2}}\leq\left(
1-\varepsilon\right)  \gamma^{T_{1}+1}<\varepsilon\gamma^{T_{1}}.
\]
In every case, $C_{2}$ has no motive to deviate from $\widehat{\sigma}^{2}$.
\end{enumerate}

Having assumed that the starting position $s_{0}$ and the strategy
profile$\left(  \widehat{\sigma}^{1},\widehat{\sigma}^{2},\widehat{\sigma}%
^{3}\right)  $ result in a capture by $C_{1}$, we have shown that no player
has a motive to change his strategy. By an analogous argument, the same holds
when $\left(  \widehat{\sigma}^{1},\widehat{\sigma}^{2},\widehat{\sigma}%
^{3}\right)  $ results in a capture by $C_{2}$ . As already mentioned,
$\left(  \widehat{\sigma}^{1},\widehat{\sigma}^{2},\widehat{\sigma}%
^{3}\right)  $ is a capturing profile, hence \emph{some} cop will capture the
robber, and no player has a motive to unilaterally change his strategy.
Consequently $\left(  \widehat{\sigma}^{1},\widehat{\sigma}^{2},\widehat
{\sigma}^{3}\right)  $ is a capturing NE\ of $\Gamma_{3}\left(  G|s_{0}%
,\gamma,\varepsilon\right)  $.
\end{proof}

\begin{remark}
\normalfont Note that in the above Theorem the NE $\left(  \widehat{\sigma
}^{1},\widehat{\sigma}^{2},\widehat{\sigma}^{3}\right)  $\ is positional.
\end{remark}

We also have the following.

\begin{theorem}
\label{prp0410}For any $G$ with $c\left(  G\right)  \geq2$, the following
holds:%
\[
\forall\left(  \gamma,\varepsilon\right)  \in\Omega^{3},\exists s_{0}\in
S:\text{there exists a non-capturing NE of }\Gamma_{3}\left(  G|s_{0}%
,\gamma,\varepsilon\right)  .
\]

\end{theorem}

\begin{proof}
Choose an $s_{0}=\left(  x,x,y,1\right)  $ of the following form: $x$ can be
any vertex of $G$ and $y$ is such that, when the one-cop CR\ is started from
$s_{0}^{\prime}=\left(  x,y,1\right)  $, the robber can avoid capture (such an
$s_{0}$ will always exist, since $c\left(  G\right)  \geq2$). The strategies
are chosen as follows.

\begin{enumerate}
\item $R$'s strategy $\widehat{\sigma}^{3}$ is the following:\ 

\begin{enumerate}
\item as long as $C_{1},C_{2}$ stay in place $R$ also stays in place;

\item if at some time $C_{1}$ (resp. $C_{2}$) is the first cop to move, $R$
starts playing an optimal one-cop CR strategy which corresponds to the $C_{1}$
(resp. $C_{2}$)\ moves.
\end{enumerate}

\item $C_{1}$'s strategy $\widetilde{\sigma}^{1}\ $is defined as follows:

\begin{enumerate}
\item if $C_{1}$ and $C_{2}$ are in the same vertex, $C_{1}$ stays in place;

\item if $C_{1}$ and $C_{2}$ are in different vertices, $C_{1}$ moves in a
shortest path towards $C_{2}$.
\end{enumerate}

\item \noindent$C_{2}$'s strategy $\widetilde{\sigma}^{2}\ $is the same as
$\widetilde{\sigma}^{1}$, with the roles of $C_{1}$ and $C_{2}$ interchanged.
\end{enumerate}

\noindent We now show that $\left(  \widetilde{\sigma}^{1},\widetilde{\sigma
}^{2},\widehat{\sigma}^{3}\right)  $ is a non-capturing NE of $\Gamma
_{3}\left(  G|s_{0},\gamma,\varepsilon\right)  $. First, since $C_{1}$ and
$C_{2}$ start at the same vertex $x$, by $\widetilde{\sigma}^{1}%
,\widetilde{\sigma}^{2}$ they will never move towards $y$; hence, under
$\left(  \widetilde{\sigma}^{1},\widetilde{\sigma}^{2},\widehat{\sigma}%
^{3}\right)  $, $R$ is not captured.

\begin{enumerate}
\item Hence $Q^{3}\left(  s_{0},\widetilde{\sigma}^{1},\widetilde{\sigma}%
^{2},\widehat{\sigma}^{3}\right)  =0$ and, clearly, $R$ cannot improve his
payoff, i.e.,%
\begin{equation}
\forall\sigma^{3}:0=Q^{3}\left(  s_{0},\widetilde{\sigma}^{1},\widetilde
{\sigma}^{2},\widehat{\sigma}^{3}\right)  \geq Q^{3}\left(  s_{0}%
,\widetilde{\sigma}^{1},\widetilde{\sigma}^{2},\sigma^{3}\right)  .
\label{eq030303}%
\end{equation}

\item Now suppose $C_{1}$ uses some $\sigma^{1}\neq\widetilde{\sigma}^{1}$ by
which, at the start of the game, he moves to some $x^{\prime}$ neighbor of
$x$. However, immediately afterwards $C_{2}$ moves by $\widetilde{\sigma}^{2}$
to the same $x^{\prime}$. In other words, $C_{1}$ and $C_{2}$ essentially move
as one cop and, since $c\left(  G\right)  \geq2$ and $R$ plays optimally,
capture will never occur. Hence
\begin{equation}
\forall\sigma^{1}:0=Q^{1}\left(  s_{0},\widetilde{\sigma}^{1},\widetilde
{\sigma}^{2},\widehat{\sigma}^{3}\right)  \geq Q^{1}\left(  s_{0},\sigma
^{1},\widetilde{\sigma}^{2},\widehat{\sigma}^{3}\right)  =0. \label{eq030304}%
\end{equation}

\item The case of $C_{2}$ is similar, but attention must be paid to some
details. Suppose $C_{2}$ uses some $\sigma^{2}\neq\widetilde{\sigma}^{2}$ by
which his frst nontrivial move is to some $x^{\prime}$ neighbor of $x$. After
him moves $R$ and, since he plays optimally, he will never move into a
\textquotedblleft vulnerable\textquotedblright\ position; in particular he
will not move into $x\in N\left[  x^{\prime}\right]  $, since then $C_{2}$
could capture him in CR; hence $R$ will never run into $C_{1}$; neither will
he be captured by $C_{2}$, since he plays optimally. So%
\begin{equation}
\forall\sigma^{2}:0=Q^{2}\left(  s_{0},\widehat{\sigma}^{1},\widehat{\sigma
}^{2},\widehat{\sigma}^{3}\right)  \geq Q^{2}\left(  s_{0},\widehat{\sigma
}^{1},\sigma^{2},\widehat{\sigma}^{3}\right)  =0. \label{eq030305}%
\end{equation}

\end{enumerate}

Combining (\ref{eq030303})-(\ref{eq030305}) we see that $\widehat{\sigma
}=\left(  \widehat{\sigma}^{1},\widehat{\sigma}^{2},\widehat{\sigma}%
^{3}\right)  $ is a non-capturing NE of $\Gamma_{3}\left(  G|s_{0}\right)  $.
\end{proof}

The above result is rather surprising when $G\ $has $c\left(  G\right)
=2$:\ while \emph{in CR }played on $G$ two optimally playing (and cooperating)
cops always capture the robber, \emph{in SCAR}\ played on the same graph there
exist non-capturing NE (even when $\varepsilon=\frac{1}{2}$, the cops'
interests coincide and they have the motive to  cooperate fully).

On the other hand, the result is \emph{not} suprprising when applied to $G$'s
with $c\left(  G\right)  \geq3$. In fact, in this case Theorem \ref{prp0410}
can be strengthened significantly: there will always exist some state with
\emph{only} non-capturing NE\footnote{However, we still have initial positions
with capturing NE; e.g., when all players start at the same vertex.}.

\begin{theorem}
\label{prp0411}For any $G$ with $c\left(  G\right)  \geq3$ the following
holds:%
\[
\forall\left(  \gamma,\varepsilon\right)  \in\Omega^{3},\exists s_{0}\in
S:\text{every NE of }\Gamma_{3}\left(  G|s_{0},\gamma,\varepsilon\right)
\text{ is non-capturing.}%
\]

\end{theorem}

\begin{proof}
Choose an $s_{0}=\left(  x,y,z,1\right)  $ such that in the two-cop
CR\ started from $s_{0}$ the robber can avoid capture; this can always be
achieved, since $c\left(  G\right)  \geq3$, provided $R$ uses an optimal (in
two-cop CR) strategy $\widehat{\sigma}^{3}$. Also take any cop strategies
$\sigma^{1}$, $\sigma^{2}$. Then the profile $\left(  \sigma^{1},\sigma
^{2},\widehat{\sigma}^{3}\right)  $ will not result in capture, in either
two-cop CR or in $\Gamma_{3}\left(  G|s_{0}\right)  $. The $\Gamma_{3}\left(
G|s_{0}\right)  $ payoffs will be%
\[
\forall\sigma^{1},\sigma^{2},\quad\forall n\in\left\{  1,2,3\right\}
:Q^{n}\left(  s_{0},\sigma^{1},\sigma^{2},\widehat{\sigma}^{3}\right)  =0.
\]
Clearly, no player can improve his payoff by unilaterally changing his
strategy. Hence, for every $\sigma^{1},\sigma^{2}$, $\left(  \sigma^{1}%
,\sigma^{2},\widehat{\sigma}^{3}\right)  $ is a non-capturing NE in
$\Gamma_{3}\left(  G|s_{0}\right)  $. On the other hand, take \emph{any} NE
$\left(  \sigma^{1},\sigma^{2},\sigma^{3}\right)  $ of $\Gamma_{3}\left(
G|s_{0}\right)  $; then we must have
\[
Q^{3}\left(  s_{0},\sigma^{1},\sigma^{2},\sigma^{3}\right)  =0
\]
because otherwise $R$ could use $\widehat{\sigma}^{3}$ and unilaterally
improve his payoff. Hence \emph{every} NE $\left(  \sigma^{1},\sigma
^{2},\sigma^{3}\right)  $ of $\Gamma_{3}\left(  G|s_{0}\right)  $ is non-capturing.
\end{proof}

The following corollary illuminates the connection of capturing and
non-capturing NE of $\Gamma_{3}\left(  G|s_{0},\gamma,\varepsilon\right)  $ to
the classic cop number. The first part of the corollary is obtained from
Theorem \ref{prp0410}; the second from Theorems \ref{prp0405} and
\ref{prp0407}.

\begin{corollary}
\label{prp0412}Given a graph $G$:

\begin{enumerate}
\item suppose that for all $\left(  \gamma,\varepsilon\right)  \in\Omega^{3}$
and $s_{0}\in S$, every NE of $\Gamma_{3}\left(  G|s_{0},\gamma,\varepsilon
\right)  $ is capturing; then $c\left(  G\right)  =1$.

\item suppose that for all $\left(  \gamma,\varepsilon\right)  \in\Omega^{3}$
there exists some $s_{0}\in S$ such that every NE of $\Gamma_{3}\left(
G|s_{0},\gamma,\varepsilon\right)  $ is non-capturing; then $c\left(
G\right)  \geq3$.
\end{enumerate}
\end{corollary}

Finally, combining Theorem \ref{prp0405} \ and the first part of Corollary
\ref{prp0412} we get the following.

\begin{corollary}
\label{prp0414}$G$ is cop-win iff : for all $\left(  \gamma,\varepsilon
\right)  \in\Omega^{3}$ and $s_{0}\in S$, every NE\ of $\Gamma_{3}\left(
G|s_{0},\gamma,\varepsilon\right)  $ is capturing.
\end{corollary}

\section{$N$-Player SCAR\label{sec05}}

\subsection{Preliminaries\label{sec0501}}

The generalization of $\Gamma_{3}\left(  G\right)  $ to $\Gamma_{N}\left(
G\right)  $, i.e., the $N$-player SCAR game is straightforward. For any
$N\geq2$, $\Gamma_{N}\left(  G\right)  $ is played by $N-1$ cops (denoted by
$C_{1},...,C_{N-1}$) and a robber (denoted by $R$)\ who move along the edges
of $G$. \footnote{Note that the case $N=2$, i.e., one cop vs. one robber, is
also included in the formulation.} The game starts from a prescribed initial
position $s_{0}$ and is played in turns, with a single player moving at every
turn; the moving sequence is $...\rightarrow C_{1}\rightarrow C_{2}%
\rightarrow...\rightarrow C_{N-1}\rightarrow R\rightarrow...$. The following
briefly presented quantities are straightforward generalizations of those
defined in Section \ref{sec02}.

The player set is $I=\left\{  1,2,...,N\right\}  $ or $I=\left\{  C_{1}%
,C_{2},...,C_{N-1},R\right\}  $. A game position or state has the form
$s=\left(  x^{1},x^{2},...,x^{N},p\right)  $, where $x^{n}$ denotes the
position of the $n$-th player and $p$ denotes the player who has the next
move. For $n\in\left\{  1,2,...,N\right\}  $, set
\[
S^{n}=\left\{  s=\left(  x^{1},...,x^{N},n\right)  :(x^{1},...,x^{N})\in
V^{N}\text{ and }n\in I\right\}
\]
is the set of states where player $n$ has the next move. The set $S$ of all
states of the game is%
\[
S=S^{1}\cup S^{2}\cup...\cup S^{N}\cup\left\{  \tau\right\}  \text{,}%
\]
where $\tau$ is as before the terminal state and the set $S^{\prime
}=S\backslash\left\{  \tau\right\}  $ is the set of non-terminal states. We
also define the set $S_{C}$ of capture states and the set $S_{NC}$ of
non-capture states as follows:%
\begin{align*}
S_{C}  &  :=\{s=(x^{1},x^{2},...,x^{N},n)\in S^{\prime}:\exists i\in
\{1,2,...,N-1\}:x^{i}=x^{N}\},\\
S_{NC}  &  :=\{s=(x^{1},x^{2},...,x^{N},n)\in S^{\prime}:\forall
i\in\{1,2,...,N-1\}\text{ }x^{i}\neq x^{N}\text{ }\}\text{.}%
\end{align*}
An alternative partition of $S$ is therefore%
\[
S=S_{C}\cup S_{NC}\cup\left\{  \tau\right\}  \text{.}%
\]
Moreover, and since we can have simultaneous captures by any subset
$\{n_{1},n_{2},...\}$ of $\{1,2,...,$ $N-1\}$, we define sets $S_{C}^{n_{1}},$
$S_{C}^{n_{1}n_{2}},...,$ $S_{C}^{12...N-1}$ analogously to the sets
$S_{C}^{1}$, $S_{C}^{2}$, $S_{C}^{12}$; the union of all these sets is of
course $S_{C}$.

Action sets $A^{n}(s)$ and the transition function $\mathbf{T}\left(
s,a\right)  $ are defined in a similar fashion as in Section \ref{sec02}. In
case at time $t$ the state is $s_{t}=\left(  x_{t}^{1},x_{t}^{2},...,x_{t}%
^{N},n\right)  \in S^{n}$ and at time $t+1$ the move by player $n$ is
$a_{t+1}^{n}\in A^{n}(s_{t})$, then we use again the shorthand
\[
s_{t+1}=\mathbf{T}\left(  s_{t},a_{t+1}^{n}\right)  \text{.}%
\]

Capture time $T_{C}$, histories and strategies are also defined analogously to
Section \ref{sec02}. The same is true for payoffs. Specifically, at every
non-capture state $s_{t}\in S_{NC}\cup\left\{  \tau\right\}  $, the immediate
reward to each player $q^{n}(s_{t})$ is zero; at every state $s_{t}\in
S_{C}^{n_{1}...n_{N_{1}}}$ (i.e., when capture is effected by $N_{1}$ cops)
the robber's loss is $\varepsilon$ and this is distributed between the $N-1$
cops as follows.

\begin{enumerate}
\item When $1\leq N_{1}\leq N-2$: each capturing (resp. non-capturing) cop
receives an immediate reward of $\frac{1-\varepsilon}{N_{1}}$ (resp.
$\frac{\varepsilon}{N-N_{1}-1}$).

\item When $N_{1}=N-1$: all cops are capturing and each receives an immediate
reward of $\frac{1}{N-1}$.
\end{enumerate}

\noindent The \emph{total} payoff of player $n$ is $Q^{n}\left(  s_{0}%
,s_{1},s_{2},...\right)  =\sum_{t=0}^{\infty}\gamma^{t}q^{n}\left(
s_{t}\right)  $. The $\left(  \gamma,\varepsilon\right)  $ sets now are
\begin{align*}
\Omega^{N}  &  =\left\{  \left(  \gamma,\varepsilon\right)  :\gamma\in\left(
0,1\right)  ,\varepsilon\in\left[  0,\frac{1}{N-1}\right]  \right\}  =\left(
0,1\right)  \times\left[  0,\frac{1}{N-1}\right]  ,\\
\widetilde{\Omega}^{N}  &  =\left\{  \left(  \gamma,\varepsilon\right)
:\gamma\in\left(  0,1\right)  ,\varepsilon\in\left[  0,\frac{1}{N-1}\right]
,\gamma<\frac{\varepsilon}{1-\varepsilon}\right\}  .
\end{align*}
The choice $\varepsilon\in\left[  0,\frac{1}{N-1}\right]  $ ensures
satisfaction of the intuitive requirement that capturing cops should get at
least as much as non-capturing ones:
\begin{align*}
\left.  \varepsilon\leq\frac{1}{N-1}\right.   &  \Rightarrow\left(  \forall
N_{1}\in\left\{  1,2,...,N-2\right\}  :\varepsilon\leq\frac{N-1-\left(
N-2\right)  }{N-1}\leq\frac{N-1-N_{1}}{N-1}\right) \\
&  \Rightarrow\left(  \forall N_{1}\in\left\{  1,2,...,N-2\right\}
:\varepsilon\left(  N-1\right)  \leq N-1-N_{1}\right) \\
&  \Rightarrow\left(  \forall N_{1}\in\left\{  1,2,...,N-2\right\}
:\varepsilon N_{1}\leq\left(  1-\varepsilon\right)  \left(  N-1-N_{1}\right)
\right) \\
&  \Rightarrow\left(  \forall N_{1}\in\left\{  1,2,...,N-2\right\}
:\frac{\varepsilon}{N-1-N_{1}}\leq\frac{1-\varepsilon}{N_{1}}\right)  .
\end{align*}
In addition, again agreeing with our intuition, each capturing cop's reward is
a decreasing function of $N_{1}$. Indeed, when $1\leq N_{1}\leq N-2$, their
reward is $\frac{1-\varepsilon}{N_{1}}$ which is decreasing in $N_{1}$, with
minimum value achieved at $N_{1}=N-2$ and equal to $\frac{1-\varepsilon}{N-2}%
$; and when $N_{1}=N-1$ (all cops are capturing)\ we have:
\[
\varepsilon\leq\frac{1}{N-1}\Rightarrow\frac{1-\varepsilon}{N-2}\geq
\frac{1-\frac{1}{N-1}}{N-2}=\frac{1}{N-1}.
\]
In short, the fewer capturing cops we have, the more is each of them rewarded.

It is however worth noting that under Nash equilibrium we will always have a
single capturing cop: the only way to have multiple capturing cops is if the
robber moves into a cop-occupied vertex. This he will never do in NE (where he
wants to maximize capture time) since he can always postpone capture by
staying in place.

\subsection{Existence of NE\label{sec0502}}

The next theorem shows the existence of positional NE\ for every $\Gamma
_{N}\left(  G|s_{0}\right)  $. It generalizes Theorem \ref{prp0401} and is
proved very similarly; hence the proof is omitted.

\begin{theorem}
\label{prp0501}For every graph $G$ and for every $s_{0}\in S$, $\left(
\gamma,\varepsilon\right)  \in\Omega^{N}$ the game $\Gamma_{N}\left(
G|s_{0},\gamma,\varepsilon\right)  $ has deterministic positional NE.
Specifically, there exists a deterministic positional profile $\sigma_{\ast
}=\left(  \sigma_{\ast}^{1},\sigma_{\ast}^{2},...,\sigma_{\ast}^{N}\right)  $
such that%
\begin{equation}
\forall n,\forall s_{0},\forall\sigma^{n}:Q^{n}\left(  s_{0},\sigma_{\ast}%
^{n},\sigma_{\ast}^{-n}\right)  \geq Q^{n}\left(  s_{0},\sigma^{n}%
,\sigma_{\ast}^{-n}\right)  .
\end{equation}
For every $s$ and $n$, let $u^{n}\left(  s\right)  =Q^{n}\left(
s,\sigma_{\ast}\right)  $. Then the \ following equations are satisfied%
\begin{align}
\forall n,\forall s  &  \in S^{n}:\sigma_{\ast}^{n}\left(  s\right)  =\arg
\max_{a^{n}\in A^{n}\left(  s\right)  }\left[  q^{n}\left(  s\right)  +\gamma
u^{n}\left(  \mathbf{T}\left(  s,a^{n}\right)  \right)  \right]  ,\\
\forall n,m,\forall s  &  \in S^{n}:u^{m}\left(  s\right)  =q^{m}\left(
s\right)  +\gamma u^{m}\left(  \mathbf{T}\left(  s,\sigma_{\ast}^{n}\left(
s\right)  \right)  \right)  .
\end{align}

\end{theorem}

The next theorem generalizes Theorem \ref{prp0404} and shows that every
$\Gamma_{N}\left(  G|s_{0}\right)  $ has NE which are not positional. The
proof (which is similar to that of Theorem \ref{prp0404} and hence will be
omitted) depends on auxiliary two-player zero-sum games $\Gamma_{N}^{1}\left(
G|s_{0}\right)  $, ..., $\Gamma_{N}^{N}\left(  G|s_{0}\right)  $, where
$\Gamma_{N}^{n}\left(  G|s_{0}\right)  $ is the two-player game with initial
state $s_{0}$ in which $P_{n}$ \ (who has payoff $Q^{n}$) plays against
$P_{-n}$ (who has payoff $-Q^{n}$ and controls $\left\{  1,2,...,N\right\}
\backslash\left\{  n\right\}  $). Similarly to the 3-player case, for each
$s\in S\ $and $n\in\left\{  1,2,...,N\right\}  $, the game $\Gamma_{N}%
^{n}\left(  G|s_{0}\right)  $ has a value and the players have optimal
deterministic positional strategies. Strategies $\phi_{n}^{n}$ and $\phi
_{n}^{m}$ are as in Section \ref{sec04}. The threat strategy of the $n$-th
player in the $N$-player game $\Gamma_{N}\left(  G|s_{0}\right)  $ is
$\overline{\sigma}^{n}$, defined (exactly as in Section \ref{sec04}) as
follows:%
\begin{equation}
\overline{\sigma}^{n}=\left\{
\begin{array}
[c]{ll}%
\phi_{n}^{n} & \text{as long as every player }m\in\left\{  1,2,...,N\right\}
\backslash n\text{ follows }\phi_{m}^{m}\text{; }\\
\phi_{m}^{n} & \text{as soon as some player }m\in\left\{  1,2,...,N\right\}
\backslash n\ \text{ \textquotedblleft deviates\textquotedblright\ from }%
\phi_{m}^{m}\text{.}%
\end{array}
\right.  \label{eqthrstrat02}%
\end{equation}

\noindent Keeping the above in mind, we can prove the following.

\begin{theorem}
\label{prp0502}For every graph $G$ and for every $N\geq3$, $\ s_{0}\in S$,
$\left(  \gamma,\varepsilon\right)  \in\Omega^{N}$, in the game $\Gamma
_{N}\left(  G|s_{0},\gamma,\varepsilon\right)  $ we have%
\begin{equation}
\forall n\in\left\{  1,2,...,N\right\}  ,\forall\sigma^{n}:Q^{n}\left(
s_{0},\overline{\sigma}^{1},\overline{\sigma}^{2},...,\overline{\sigma}%
^{N}\right)  \geq Q^{n}(s,\sigma^{n},\overline{\sigma}^{-n})
\end{equation}
where $\overline{\sigma}^{n}$ (for $n\in\left\{  1,2,...,N\right\}  $) is a
deterministic non-positional strategy of the form (\ref{eqthrstrat02}).
\end{theorem}

As in 3-player SCAR $\Gamma_{3}\left(  G|s_{0}\right)  $, the above hold for
any $\varepsilon\in\mathbb{[}0,1]$, not just for $\varepsilon\in\left[
0,\frac{1}{N-1}\right]  $.

\subsection{Cop Number, Capturing and Non-capturing NE\label{sec0503}}

The following results generalize those appearing in Section \ref{sec0403} and
hold for every $N\geq2$.

\begin{theorem}
\label{prp0503}For any $G$ with $c\left(  G\right)  =1$ the following holds:%
\[
\forall\left(  \gamma,\varepsilon\right)  \in\Omega^{N},\forall s_{0}\in
S:\text{every NE of }\Gamma_{N}\left(  G|s_{0},\gamma,\varepsilon\right)
\text{ is capturing. }%
\]

\end{theorem}

\begin{theorem}
\label{prp0504}For any $G$ with $c\left(  G\right)  \leq N-1$ the following
holds:%
\[
\forall\left(  \gamma,\varepsilon\right)  \in\Omega^{N},\forall s_{0}\in
S:\text{there exists a capturing NE of }\Gamma_{N}\left(  G|s_{0}%
,\gamma,\varepsilon\right)  .
\]

\end{theorem}

\begin{theorem}
\label{prp0505}For any $G$ with $c\left(  G\right)  \leq N-1$, let
$\widehat{\sigma}=\left(  \widehat{\sigma}^{1},\widehat{\sigma}^{2}%
,...,\widehat{\sigma}^{N}\right)  $ be a strategy profile which is optimal in
the $\left(  N-1\right)  $-cop CR game. Then the following holds:%
\[
\forall\left(  \gamma,\varepsilon\right)  \in\widetilde{\Omega}^{N},\forall
s_{0}\in S:\widehat{\sigma}\text{ is a capturing NE of }\Gamma_{N}\left(
G|s_{0},\gamma,\varepsilon\right)  .
\]

\end{theorem}

\begin{theorem}
\label{prp0506}For any $G$ with $c\left(  G\right)  \geq2$, the following
holds:%
\[
\forall\left(  \gamma,\varepsilon\right)  \in\Omega^{N},\exists s_{0}\in
S:\text{there exists a non-capturing NE of }\Gamma_{N}\left(  G|s_{0}%
,\gamma,\varepsilon\right)  .
\]

\end{theorem}

\begin{theorem}
\label{prp0507}For any $G$ with $c\left(  G\right)  \geq N$ the following
holds:%
\[
\forall\left(  \gamma,\varepsilon\right)  \in\Omega^{N},\exists s_{0}\in
S:\text{every NE of }\Gamma_{N}\left(  G|s_{0},\gamma,\varepsilon\right)
\text{ is non-capturing.}%
\]

\end{theorem}

\begin{corollary}
\label{prp0508}Given a graph $G$:

\begin{enumerate}
\item suppose that for all $\left(  \gamma,\varepsilon\right)  \in\Omega^{N}$
and $s_{0}\in S$, every NE of $\Gamma_{N}\left(  G|s_{0},\gamma,\varepsilon
\right)  $ is capturing; then $c\left(  G\right)  =1$.

\item suppose that for all $\left(  \gamma,\varepsilon\right)  \in\Omega^{N}$
there exists some $s_{0}\in S$ such that every NE of $\Gamma_{N}\left(
G|s_{0},\gamma,\varepsilon\right)  $ is non-capturing; then $c\left(
G\right)  \geq N$.
\end{enumerate}
\end{corollary}

\begin{corollary}
\label{prp0509}$G$ is cop-win iff : for all $\left(  \gamma,\varepsilon
\right)  \in\Omega^{N}$ and $s_{0}\in S$, every NE\ of $\Gamma_{N}\left(
G|s_{0},\gamma,\varepsilon\right)  $ is capturing.
\end{corollary}

\subsection{Selfish Cop Number\label{sec0504}}

We know that the cop number $c\left(  G\right)  $ of graph $G$ is the minimum
number of cops required to guarantee (when the cops play optimally and for any
robber strategy and starting position) capture in CR\ played on $G$. Define
correspondingly the \emph{selfish cop} \emph{number} for the SCAR game.

\begin{definition}
\label{prp0510}The \emph{selfish cop number} of a graph $G$ is denoted by
$c_{s}\left(  G\right)  $ and defined to be the smallest $K$ such that: for
any $\left(  \gamma,\varepsilon\right)  \in\Omega^{K+1}$ and any $s_{0}\in S$,
there exists a capturing NE of $\Gamma_{K+1}\left(  G|s_{0},\gamma
,\varepsilon\right)  $.
\end{definition}

The selfish cop number equals the classic one, as demonstrated in the following.

\begin{theorem}
\label{prp0511}For every graph $G$ we have $c_{s}\left(  G\right)  =c\left(
G\right)  $.
\end{theorem}

\begin{proof}
Take any $K$ such that $K\geq c\left(  G\right)  $, then by Theorem
\ref{prp0504} we have that, for every $\left(  \gamma,\varepsilon\right)
\in\Omega^{K+1}$ and every $s_{0}\in S$ there exists a capturing NE of
$\Gamma_{K+1}\left(  G|s_{0},\gamma,\varepsilon\right)  $. On the other hand,
take any $K\leq c\left(  G\right)  -1$, then by Theorem \ref{prp0507} we have
that, for every $\left(  \gamma,\varepsilon\right)  \in\Omega^{K+1}$ there
exists some $s_{0}\in S$, such that there exists no capturing NE of
$\Gamma_{K+1}\left(  G|s_{0},\gamma,\varepsilon\right)  $. Hence $c_{s}\left(
G\right)  $ (i.e., the smallest $K$ such that for every $\left(
\gamma,\varepsilon\right)  \in\Omega^{K+1}$ and every $s_{0}\in S$ there
exists a capturing NE of $\Gamma_{K+1}\left(  G|s_{0},\gamma,\varepsilon
\right)  $) equals $c\left(  G\right)  $.
\end{proof}

\subsection{A Connection between CR and SCAR\label{sec0505}}

We will now show that a slightly modified version of $N$-player SCAR is, in a
certain sense, equivalent to the CR game with $N-1$ cops. \ The modification
consists in letting $\varepsilon$ be a function of $N$ and $N_{1}$, namely we
use $\varepsilon\left(  N,N_{1}\right)  =\frac{N-1-N_{1}}{N-1}$. We will
denote the modified SCAR\ game by $\Gamma_{N}\left(  G|s_{0},\gamma
,\varepsilon\left(  N,N_{1}\right)  \right)  $.

The modification implies that the payoff of each capturing state depends on
the number of capturing cops. The distribution of the payoff remains the same, i.e.,

\begin{enumerate}
\item each capturing cop receives a reward of $\frac{1-\varepsilon\left(
N,N_{1}\right)  }{N_{1}}=\frac{1-\frac{N-1-N_{1}}{N-1}}{N_{1}}=\allowbreak
\frac{1}{N-1}$;

\item each non-capturing cop receives a reward of $\frac{\varepsilon\left(
N,N_{1}\right)  }{N-1-N_{1}}=\frac{\frac{N-1-N_{1}}{N-1}}{N-1-N_{1}%
}=\allowbreak\frac{1}{N-1}$.
\end{enumerate}

\noindent In other words, for every capture each cop (whether he is capturing
or non-capturing)\ receives the same reward.

SCAR with the above modification of $\varepsilon$ falls under the general
formulation of discounted stochastic games and all our previous results still
hold. Furthermore, the $\left(  N-1\right)  $-cops CR game is
\emph{payoff-equivalent} to $\Gamma_{N}\left(  G|s_{0},\gamma,\varepsilon
\left(  N,N_{1}\right)  \right)  $, by which we mean the following. Take any
strategies $\sigma^{1},\sigma^{2},...,\sigma^{N}$ and apply them

\begin{enumerate}
\item to $\Gamma_{N}\left(  G|s_{0},\gamma,\varepsilon\left(  N,N_{1}\right)
\right)  $, with $\sigma^{n}$ being the strategy of the $n$-th player;

\item to the $\left(  N-1\right)  $-cops CR game, with $\sigma^{n}$\ (for
$n\in\left\{  1,2,...,N-1\right\}  $)\ being the strategy the cop player uses
for his $n$-th token, and $\sigma^{N}$ being the strategy the robber player uses.
\end{enumerate}

\noindent Then the same history $(s_{0},s_{1},s_{2},...)$ will be produced in
the two games (they are path-equivalent\noindent)\ and, furthermore, the sum
of the total payoffs of the cops (resp. robber)\ in $\Gamma_{N}\left(
G|s_{0},\gamma,\varepsilon\left(  N,N_{1}\right)  \right)  $ will be the same
as the payoff of the cops (resp. robber)\ in the $\left(  N-1\right)  $-cops
CR game.

Since all cops in $\Gamma_{N}\left(  G|s_{0},\gamma,\varepsilon\left(
N,N_{1}\right)  \right)  $ receive the same payoff, their interests totally
coincide:\ they all want \emph{some} cop to capture the robber in the shortest
possible time, just like in the $\left(  N-1\right)  $-cops CR game. Hence the
following catn be proved in a similar way as Theorems \ref{prp0409} and
\ref{prp0505}.

\begin{theorem}
\label{prp0512}If the profile $\left(  \left(  \widehat{\sigma}_{1}%
,...,\widehat{\sigma}_{N-1}\right)  ,\widehat{\sigma}_{N}\right)  $ is optimal
in the $\left(  N-1\right)  $-cops CR game, then the profile $\left(
\widehat{\sigma}_{1},...,\widehat{\sigma}_{N-1},\widehat{\sigma}_{N}\right)  $
is a NE\ of $\Gamma_{N}\left(  G|s_{0},\gamma,\varepsilon\left(
N,N_{1}\right)  \right)  $.
\end{theorem}

\section{Conclusion\label{sec06}}

As we have already mentioned, very little work has been previously done on
\emph{multi-player} pursuit games. In this sense SCAR furnishes a novel
generalization of CR and its numerous \emph{two-player} variants. We find
especially interesting the following aspects of SCAR.

\begin{enumerate}
\item On the \textquotedblleft technical\textquotedblright\ side, the
formulation of SCAR\ as a discounted game is quite advantageous. In the
\textquotedblleft natural\textquotedblright\ formulation of a pursuit game,
payoff is expected capture time; since this can be unbounded, there is no
obvious way to establish the existence of NE (in the multi-player case). On
the other hand, in the SCAR formulation payoff is a \emph{discounted} constant
(see (5));\ consequently the existence of a deterministic positional NE
follows immediately from Fink's classical result. Furthermore, because SCAR is
a perfect information game, its payoff can be immediately converted to capture
time, thus preserving the semantics of a pursuit game.

\item On the \textquotedblleft conceptual\textquotedblright\ side, our results
indicate that (perhaps surprisingly) even when $N-1$ cops can capture the
robber if they cooperate, they may settle on a non-cooperating, non-capturing
Nash equilibrium. This is somewhat similar to the \textquotedblleft
lack-of-cooperation\textquotedblright\ phenomenon observed in other branches
of Game Theory (e.g., in Prisoner's Dilemma and the Tragedy of the Commons).
\end{enumerate}

The above facts indicate further research directions, which we intend to
pursue in the future. We conclude this paper by briefly discussing some such directions.

\begin{enumerate}
\item \emph{Refinement of equilibria}. The apparent paradox of non-capturing
Nash equilibria may be resolved by using more refined equilibria concepts
(subgame perfect equilibria, strong equilibria, admissible equilbria etc.). An
obvious target then is to establish the existence and nature of such equilbria
in SCAR.

\item \emph{SCAR\ variants}. These are obtained by changing the number and /
or behaviors of the cops and robbers. Some possibilities are listed below; the
methods of the current paper can be used to study the resulting variants.

\begin{enumerate}
\item One cop pursues several selfish robbers; each robber pays a penalty if
he is captured and a (lower, perhaps zero)\ penalty if another robber is
captured. In a sense this is the dual of the game we have studied in this paper.

\item More generally, $N-M$ cops pursue $M$ robbers; the payoff of each player
may reflect a completely or partially selfish behavior on his part.

\item Even more generally, \emph{teams} of cops pursue teams of robbers; \ a
team is a set of tokens controlled by a single player.

\item The robbers can be \textquotedblleft passive\textquotedblright%
:\ (i)\ they move on the graph according to predetermined, known transition
functions and (ii)\ they do not receive any payoff (but their \ capture
results in payoffs to the capturing and non-capturing cops). Conversely, we
can have active robbers and passive cops.
\end{enumerate}

\item \emph{SCAR\ generalizations}. More generally, a family of
\emph{generalized multi-player pursuit / evasion games on graphs} can be
obtained by varying the \textquotedblleft capture
relationship\textquotedblright\ between players. Here are two examples.

\begin{enumerate}
\item A game played by players $P_{1}$, $P_{2}$, ..., $P_{N}$, in which
$P_{n}$ pursues $P_{n+1}$ (for $n\in\left\{  1,2,...,N-1\right\}  $); here we
have a \textquotedblleft linear\textquotedblright\ pursuit relationship.

\item The same as above but also $P_{N}$ pursues $P_{1}$; here we have a
\textquotedblleft cyclic\textquotedblright\ pursuit relationship.
\end{enumerate}

\noindent Hence a player will, in general, be simultaneously pursuer and
evader. Pursuit relationships are specified in terms of appropriate player
payoffs\footnote{A minimum requirement (to preserve the semantics of pursuit /
evasion) is that total payoff is nondecreasing (resp. nonincreasing)\ with
capture time for the evader (resp. pursuer).}, e.g., the capturing (resp.
captured) player receives (resp. pays)\ one time discounted unit. Again, the
resulting games can be studied by the methods of the current paper.

\item \emph{Non-perfect-information games}. A more drastic change (which can
be used in conjunction to any of the previously mentioned variations) is to
allow for \emph{simultaneous} or \emph{concurrent} player moves. This results
in non-perfect-information games and their study will require more powerful
methods than \ the ones presented in the current paper.
\end{enumerate}


\begin{thebibliography}{99}                                                                                               %


\bibitem {bonato2011}Bonato, Anthony, and Richard J. Nowakowski.
\textquotedblleft The game of cops and robbers on graphs\textquotedblright.
Vol. 61. Providence: American Mathematical Society, 2011.

\bibitem {bonato2017}Bonato, Anthony, and Gary MacGillivray. \textquotedblleft
Characterizations and algorithms for generalized Cops and Robbers
games\textquotedblright, accepted to \emph{Contributions to Discrete
Mathematics}.

\bibitem {berwanger}D. Berwanger, Graph games with perfect information,
\emph{Unpublished Manuscript}.

\bibitem {boros2009}Boros, Endre, and Vladimir Gurvich. \textquotedblleft Why
chess and back gammon can be solved in pure positional uniformly optimal
strategies\textquotedblright. Rutcor Research Report 21-2009, Rutgers
University, 2009.

\bibitem {chatterjee2003}K. Chatterjee, R. Majumdar and M. Jurdzi\'{n}ski.
\textquotedblleft On Nash equilibria in stochastic games\textquotedblright. In
\emph{International Workshop on Computer Science Logic}, pp. 26-40. Springer
Berlin Heidelberg, 2004.

\bibitem {everett1957}Everett, Hugh. "Recursive games." Contributions to the
Theory of Games 3.39 (1957): 47-78.

\bibitem {filar1996}J. Filar and K. Vrieze. \emph{Competitive Markov decision
processes}. 1996.

\bibitem {fink1963}A.M. Fink. \textquotedblleft Equilibrium in a stochastic
$n$-person game.\textquotedblright\ \emph{Journal of Science of the Hiroshima
University}, Series AI (mathematics), vol. 28, no.1 (1964), pp. 89-93.

\bibitem {foley1974}M. Foley and W. Schmitendorf. \textquotedblleft A class of
differential games with two pursuers versus one evader.\textquotedblright%
\ \emph{IEEE Transactions on Automatic Control}, vol. 19, no. 3 (1974), pp. 239-243.

\bibitem {fraenkel1996}Fraenkel, Aviezri S. "Combinatorial games: selected
bibliography with a succinct gourmet introduction." (1996).

\bibitem {isaacs1965}R. Isaacs. \emph{Differential games: a mathematical
theory with applications to warfare and pursuit, control and optimization}.
Courier Corporation, 1999.

\bibitem {kehagias2013}Kehagias, Athanasios, Dieter Mitsche, and
Pawe\l \ Pra\l at. "Cops and invisible robbers: The cost of drunkenness."
Theoretical Computer Science 481 (2013): 100-120.

\bibitem {kehagias2014}Kehagias, Athanasios, Dieter Mitsche, and
Pawe\l \ Pra\l at. "The role of visibility in pursuit/evasion games." Robotics
3.4 (2014): 371-399.

\bibitem {kehagias2016}Ath. Kehagias and G. Konstantinidis. \textquotedblleft
Selfish Cops and Passive Robber: Qualitative Games\textquotedblright.
\emph{Theoretical Computer Science}, vol.680, pp.25-35, 2017.

\bibitem {mazala2002}R. Mazala, \textquotedblleft Infinite
games.\textquotedblright\ Automata logics, and infinite games. Springer Berlin
Heidelberg, 2002. 23-38.

\bibitem {kehagias2016a}Konstantinidis, Georgios, and Ath Kehagias.
\textquotedblleft Simultaneously moving cops and robbers.\textquotedblright%
\ Theoretical Computer Science 645 (2016): 48-59.

\bibitem {mertens2002}Mertens, Jean-Fran\c{c}ois. "Stochastic games." Handbook
of game theory with economic applications 3 (2002): 1809-1832.

\bibitem {nash1950}Nash, John F. \textquotedblleft Equilibrium points in
n-person games\textquotedblright. \emph{Proceedings of the Aational Academy of
Sciences}, vol. 36 (1950): 48-49.

\bibitem {nowakowski1983}R. Nowakowski and P. Winkler. \textquotedblleft
Vertex-to-vertex pursuit in a graph\textquotedblright. \emph{Discrete
Mathematics} vol. 43, no. 2-3, pp. 235-239, 1983.

\bibitem {nowakowski1998}Nowakowski, Richard J., ed. Games of no chance. Vol.
29. Cambridge University Press, 1998.

\bibitem {quilliot1978}Quilliot, Alain. "These de 3e cycle." Universit\'{e} de
Paris VI (1978): 131-145.

\bibitem {talebi2017}Talebi, S., \& Simaan, M. A. (2017, June). Multi-pursuer
pursuit-evasion games under parameters uncertainty: A Monte Carlo approach. In
System of Systems Engineering Conference (SoSE), 2017 12th (pp. 1-6). IEEE.

\bibitem {thuijsman1997}F. Thuijsman and T.E.S. Raghavan. \textquotedblleft
Perfect information stochastic games and related classes\textquotedblright.
\emph{International Journal of Game Theory}, vol. 26, no. 3, pp. 403-408, 1997.

\bibitem {tynianski2017}Tynyanskii, N. T., and V. I. Zhukovskii. "Nonzero-sum
differential games (the noncooperative variant)." Journal of Soviet
Mathematics 12.6 (1979): 759-798.

\bibitem {vieille2002}Vieille, Nicolas. "Stochastic games: Recent results."
Handbook of game theory with economic applications 3 (2002): 1833-1850.
\end{thebibliography}
\end{document}